\documentclass[aps,pra,twocolumn,10pt,superscriptaddress,nofootinbib]{revtex4-2}

\usepackage[T1]{fontenc}
\usepackage[utf8]{inputenc}
\usepackage{lmodern}
\usepackage{amsmath,amssymb,amsfonts,amsthm,mathtools}
\usepackage{bm}
\usepackage{physics}
\usepackage{graphicx}
\usepackage{xcolor}
\usepackage{hyperref}
\usepackage{mathrsfs}
\usepackage{enumitem}
\usepackage{booktabs}

\hypersetup{
  colorlinks=true,
  linkcolor=blue,
  citecolor=blue,
  urlcolor=blue
}

\newcommand{\I}{\mathbb{I}}
\renewcommand{\Tr}{\operatorname{Tr}}
\newcommand{\diag}{\operatorname{diag}}
\newcommand{\HS}{\mathrm{HS}}

\newcommand{\hbin}{h_2}
\newcommand{\x}{\mathrm{x}}

\newcommand{\z}{\mathrm{z}}

\newtheorem{lemma}{Lemma}

\theoremstyle{remark}
\newtheorem{remark}{Remark}

\begin{document}

\title{A geometric Fano--Procrustes framework for purification-based distances and quantum channels analysis}

\author{Tristán M. Osán}
\email{tristan.osan@unc.edu.ar}
\affiliation{Instituto de F\'{i}sica Enrique Gaviola (IFEG-UNC-CONICET), Universidad Nacional de Córdoba, 
Av. Medina Allende s/n, Ciudad Universitaria. Córdoba (X5000HUA). Argentina.}
\affiliation{Universidad Nacional de Córdoba. Facultad de Matemática, Astronomía, Física y Computación. Grupo de teoría de la materia condensada. Av. Medina Allende s/n, Ciudad Universitaria. Córdoba (X5000HUA). Argentina.}

\date{\today}

\begin{abstract}
In this work we reformulate the Uhlmann 
purification-overlap optimization and develop a purification-based
geometric framework for the analysis of mixed qubit states and qubit
channels. Using the Fano representation of two-qubit pure states,
a purification is described in terms of the Bloch vector of the system,
the ancilla Bloch vector, and a real correlation matrix. For a fixed
one-qubit mixed state, the freedom in the choice of purification can 
be parametrized by proper rotations acting on the ancillary
degrees of freedom. As a result, the optimization over purifications
entering the definition of the metric \(D_N\) introduced in
Ref.~\cite{Lamberti2009} is reduced to an orthogonal Procrustes problem
on the Lie group \(SO(3)\). This reduction yields not only the maximal
purification overlap, but also the optimal rotation relating the
purification frames. From this rotation we define a purification
misalignment angle \(\Theta\), which provides geometric information not
contained in scalar fidelity-based distinguishability measures.

The formalism is applied to representative qubit channels, including
depolarizing, bit-flip, phase-flip, amplitude-damping channels, and an imperfect
quantum NOT gate. For symmetry-adapted evolutions preserving the
Bloch-vector direction, the optimal rotation is trivial and
\(\Theta=0\), whereas noncollinear channel actions generate a nonzero
misalignment. The pair \((D_N,\Theta)\) therefore separates the
magnitude of the maximal purification overlap from the geometric
reorientation of the optimal purification frames. Since the optimal
Procrustes rotation can be lifted to a local unitary acting on the
ancilla, the construction also provides an operational interpretation
of the optimal purification in terms of an ancilla-side transformation.
\end{abstract}

\maketitle

\section{Introduction}

The characterization of quantum states and their transformations is a
central problem in quantum information theory. In particular, the
quantification of distinguishability between mixed states plays a
fundamental role in quantum communication, quantum control, and the
analysis of noisy quantum processes~\cite{Nielsen2000,Schumacher2010,Bengtsson2017,Wilde2017,Watrous2018,DAlessandro2021,Wiseman2010}.
Among the various approaches to
this problem, purification-based constructions occupy a prominent
position, since they connect geometric notions of distance with
operationally meaningful quantities such as transition probabilities,
maximal overlaps, and fidelities~\cite{Uhlmann1976,Jozsa1994,Nielsen2000,Bengtsson2017}.

A natural construction of this type is obtained by extending a distance
defined on pure states to mixed states through purification~\cite{Lamberti2009,Osan2013}. Given a
density operator \(\rho\), one considers all possible purifications in
an enlarged Hilbert space and defines the distance between two mixed
states through an optimization over such purifications. For
overlap-based distances, this optimization is equivalent, by Uhlmann's
theorem~\cite{Uhlmann1976,Jozsa1994,Nielsen2000,Watrous2018}, to maximizing the overlap between purifications, leading to
expressions involving the Uhlmann fidelity. This result provides a
compact and widely used scalar characterization of mixed-state
distinguishability. However, it does not, by itself, reveal the
structure of the optimizing purifications or the geometric mechanism by
which the optimal overlap is attained.

This distinction is important. Any distance or metric that is a
function of the Uhlmann fidelity, or equivalently of the maximal
purification overlap, ultimately depends only on the scalar value of
the optimization. Examples include the Bures distance, the Bures angle, the root infidelity~\cite{Bures1969,Nielsen2000,Bengtsson2017},
and the entropic metric considered in this work~\cite{Lamberti2009,Osan2013}. These
quantities differ in the particular scalar function applied to the
maximal overlap, but they do not retain information about the optimizer
itself. Thus, a fidelity-based distance answers how distinguishable two
states are, but not how the corresponding purifications must be aligned
in order to realize the optimal overlap.

The purpose of the present work is to develop a framework that keeps
both pieces of information. We focus on the qubit case, where additional
geometric structure can be exploited. Every mixed qubit state admits a
Bloch-vector representation, and bipartite two-qubit states can be
written in Fano form in terms of local Bloch vectors and a real
correlation matrix. This representation provides a natural bridge
between the algebraic description of quantum states and a geometric
picture based on three-dimensional vectors and rotations.

We show that the maximization of the overlap between qubit
purifications can be reduced to an orthogonal Procrustes problem on the
rotation group \(SO(3)\). In this formulation, the optimization over
purifications is replaced by an optimization over proper rotations,
which can be solved analytically using the singular value decomposition
of a real matrix constructed from the Fano parameters of the two states~\cite{Horn1986,Constantinescu2002}. The
maximal value obtained from this Procrustes problem reproduces the same
overlap that appears in Uhlmann's theorem~\cite{Uhlmann1976,Jozsa1994,Nielsen2000,Watrous2018}, and therefore may be used to
evaluate any fidelity-based scalar distance. The advantage of the
Fano--Procrustes formulation is that it also identifies the optimizing
rotation \(S_\star\).

This observation is the central point of the present approach. The
Fano--Procrustes construction may be used with any distinguishability
measure whose input is the Uhlmann fidelity or the maximal purification
overlap \(g_\star\). Once the Procrustes optimization yields \(g_\star\), one may
evaluate, for example, a Bures-type distance, or
the entropic metric $D_N$ introduced in Ref.~\cite{Lamberti2009}, by applying the corresponding scalar
function to \(g_\star\). At the same time, the optimizer \(S_\star\)
contains geometric information that is discarded by purely
fidelity-based treatments. Beyond the determination of the scalar distance, the optimal
Procrustes rotation naturally leads to an additional geometric
quantity, i.e., the purification misalignment angle \(\Theta\). This angle is
defined from \(S_\star\) and encodes the relative
orientation of the optimal purification frames within a fixed
canonical gauge. While the scalar distance measures the amount of
distinguishability between two states, \(\Theta\) provides
complementary information about how the optimal purifications are
geometrically related.

To illustrate the framework, we analyze several quantum channels acting
on qubits, including depolarizing, phase-flip, bit-flip, and amplitude
damping channels, as well as a model of an imperfect quantum NOT gate.

The paper is organized as follows. In Sec.~\ref{sec:purification_concept}, we review the concept of purification for mixed quantum states.
In Sec.~\ref{sec:bassi_operational_background}, we introduce the scheme proposed in \cite{Bassi2003}, in which the purification process is given an operational interpretation in terms of a physical reference system coupled to the system of interest. In Sec.~\ref{sec:uhlmann_fidelity}, we introduce the essential aspects of the Uhlmann--Jozsa fidelity that are directly relevant to the results presented in this work. 
In Sec.~\ref{sec:DN_metric} we introduce the purification-based entropic metric $D_N$ employed in this work and discuss its relation to the
maximal purification overlap. In Sec.~\ref{sec:kraus_decomposition} we briefly introduce the
operator-sum representation for quantum processes. In Sec.~\ref{sec:quantum_channels} we briefly describe representative
noisy qubit channels, such as the depolarizing, bit-flip,
phase-flip, and amplitude damping channels, which provide examples of
radial, anisotropic, and nonunital dynamics on the Bloch sphere. We
also introduce an imperfect quantum NOT gate modeled as a
random-unitary channel describing a pulse-angle control error. In Sec.~\ref{sec:fano_procrustes} we present the Fano representation of
qubit purifications and derive the algebraic constraints imposed by
purity. In Sec.~\ref{sec:procrustes_derivation} we show how the overlap
optimization reduces to an orthogonal Procrustes problem. In
Sec.~\ref{sec:Procrustes_solution} we introduce a canonical gauge and
derive the explicit Procrustes matrix in terms of Bloch-vector parameters. In
Sec.~\ref{sec:Theta_angle} we define the purification misalignment angle
\(\Theta\). In Sec.~\ref{sec:theta_anisotropy}, we discuss the effect of on the misalignment angle
\(\Theta\). In Sec.~\ref{sec:analytical_results} we provide closed-form expressions
for the purification overlap and the associated metric \(D_N\) for
several important classes of quantum channels. The analysis focuses on
families of states adapted to the symmetry of the channel, so that the
Bloch vectors of the input and output states remain collinear.
In Sec.~\ref{sec:numerical} we explore the behavior of the metric $D_N$ and the optimal purification alignment
$\Theta$ for generic qubit states and channels.
In Sec.~\ref{sec:analysis_discussion} we analyze the results presented
in this work using the Fano--Procrustes construction and also introduce
an operational realization of the optimal purification.
Finally, in Sec.~\ref{concluding_remarks} we summarize the results and discuss possible extensions.

\section{Purification of mixed quantum states}
\label{sec:purification_concept}

Purification is a fundamental procedure that associates a \textit{mixed}
quantum state with a \textit{pure} state defined on an enlarged Hilbert
space. Specifically, given a density operator $\rho$ acting on a Hilbert
space $\mathcal{H}$, one can always construct a pure state on an extended
space $\mathcal{H} \otimes \mathcal{H}_{\rm aux}$ such that tracing out
the auxiliary space $\mathcal{H}_{\rm aux}$ reproduces $\rho$~\cite{Jozsa1994,Nielsen2000,Bengtsson2017}.

To illustrate the construction, let

\[
\rho = \sum_k p_k \ket{\psi_k}\bra{\psi_k}
\]

\noindent be a spectral decomposition of $\rho$. Introduce an auxiliary Hilbert space
$\mathcal{H}_{\rm aux}$ with dimension at least equal to the rank of $\rho$,
and choose an orthonormal basis $\{\ket{\phi_k}\}$ of $\mathcal{H}_{\rm aux}$.
Define the normalized state

\begin{equation*}
\ket{\Psi} = \sum_k \sqrt{p_k}\, \ket{\psi_k} \otimes \ket{\phi_k}.
\end{equation*}

Let $\rho_\Psi = \ketbra{\Psi}{\Psi}$ denote the density operator of this
pure state on $\mathcal{H} \otimes \mathcal{H}_{\rm aux}$. Taking the partial
trace over the auxiliary system yields

\begin{equation*}
\Tr_{\rm aux}(\rho_\Psi)
=
\sum_k p_k \ket{\psi_k}\bra{\psi_k}
=
\rho,
\end{equation*}

\noindent where $\Tr_{\rm aux}$ denotes the partial trace over $\mathcal{H}_{\rm aux}$.

Thus, any mixed state can be purified by embedding it into an extended
Hilbert space~\cite{Jozsa1994,Nielsen2000}. The purification of a mixed state is
not unique: a given density operator admits infinitely many purifications,
related by unitary transformations acting on the auxiliary space.
Moreover, purification is not merely a formal construction but has direct
physical meaning; for instance, Ref.~\cite{Bassi2003} discusses physical
realizations of purifications associated with specific ensembles.

In particular, all purifications of a given state are related by unitary
transformations acting solely on the auxiliary Hilbert space.

\section{Operational background: ensemble purification}
\label{sec:bassi_operational_background}

Purification is often introduced as a formal way of representing a mixed
state as the reduced state of a pure state in a larger Hilbert space.
However, Bassi and Ghirardi~\cite{Bassi2003} emphasized that
purification may also be understood operationally, in terms of a physical
reference system coupled to the system of interest.

Let

\[
\rho=\sum_i d_i\ket{\phi_i}\bra{\phi_i}
\]

be the spectral decomposition of a density operator. Introducing an
ancillary system \(K\) with an orthonormal set
\(\{\ket{a_i}\}\), one obtains the purification

\begin{equation}
\ket{\Psi}
=
\sum_i \sqrt{d_i}\,
\ket{\phi_i}\otimes\ket{a_i},
\label{eq:bassi_purification_background}
\end{equation}

\noindent which satisfies

\begin{equation*}
\Tr_K\ket{\Psi}\bra{\Psi}=\rho.
\end{equation*}

The key operational point is that different measurements on the
ancillary system generate different ensemble decompositions of the same
density operator. Indeed, if another orthonormal basis
\(\{\ket{b_j}\}\) of the ancilla is related to \(\{\ket{a_i}\}\) by a
unitary transformation,

\begin{equation*}
\ket{a_i}=\sum_j U_{ij}\ket{b_j},
\end{equation*}

\noindent then Eq.~\eqref{eq:bassi_purification_background} can be rewritten as

\begin{equation*}
\ket{\Psi}
=
\sum_j
\ket{\widetilde\chi_j}\otimes\ket{b_j},
\qquad
\ket{\widetilde\chi_j}
=
\sum_i \sqrt{d_i}\,U_{ij}\ket{\phi_i}.
\end{equation*}

A measurement on the ancilla in the basis \(\{\ket{b_j}\}\) prepares the
system in the normalized states

\[
\ket{\chi_j}
=
\frac{\ket{\widetilde\chi_j}}
{\|\ket{\widetilde\chi_j}\|}
\]

with probabilities

\[
p_j=\|\ket{\widetilde\chi_j}\|^2.
\]

Thus, by changing the measurement basis of the reference system, one
obtains different ensembles

\[
\{p_j,\ket{\chi_j}\}
\]

which all represent the same density operator \(\rho\).

Bassi and Ghirardi further showed that this construction can be given a
physical realization. One may introduce an interaction between the
system and the reference system such that

\begin{equation*}
\ket{\phi_i}\otimes\ket{a_0}
\longmapsto
\ket{\phi_i}\otimes\ket{a_i},
\qquad
\braket{a_i}{a_j}=\delta_{ij}.
\end{equation*}

By linearity, this interaction transforms

\[
\sum_i\sqrt{d_i}\ket{\phi_i}\otimes\ket{a_0}
\]

\noindent into the purified state

\[
\sum_i\sqrt{d_i}\ket{\phi_i}\otimes\ket{a_i}.
\]

For qubits, this type of purification can be implemented using standard
controlled operations, such as a controlled-NOT gate in a suitable basis.

This operational viewpoint is relevant for the present work because it
shows that the ancillary degrees of freedom appearing in a purification
can be regarded as controllable physical degrees of freedom. In the
Fano--Procrustes framework developed below, this ancillary freedom is
represented geometrically by rotations acting on the Fano parameters of the
purification. The Procrustes optimization then selects, among these
ancilla-side transformations, the one that produces the purification
with maximal overlap. Thus, the operational purification scheme provides
a useful conceptual background for the later interpretation of the
optimal Procrustes rotation as a physically meaningful ancilla operation.

\section{Uhlmann-Jozsa fidelity}
\label{sec:uhlmann_fidelity}

Let \(\rho\) and \(\sigma\) be density operators acting on a finite-dimensional
Hilbert space \(\mathcal H\). The Uhlmann fidelity is defined as~\cite{Uhlmann1976,Jozsa1994}

\begin{equation*}
F(\rho,\sigma)
=
\left[
\Tr\sqrt{\sqrt{\rho}\,\sigma\,\sqrt{\rho}}
\right]^2 .
\end{equation*}

Equivalently, one may introduce the fidelity amplitude

\begin{equation*}
\mathcal F(\rho,\sigma)
=
\Tr\sqrt{\sqrt{\rho}\,\sigma\,\sqrt{\rho}},
\end{equation*}

\noindent so that

\begin{equation*}
F(\rho,\sigma)=\mathcal F(\rho,\sigma)^2 .
\end{equation*}

Uhlmann's theorem~\cite{Uhlmann1976,Jozsa1994,Nielsen2000,Watrous2018} states that the fidelity amplitude is equal to the
maximal overlap between purifications of the two states:

\begin{equation}
\mathcal F(\rho,\sigma)
=
\max_{\ket{\Psi_\rho},\,\ket{\Psi_\sigma}}
\left|
\braket{\Psi_\rho}{\Psi_\sigma}
\right|,
\label{eq:uhlmann_theorem}
\end{equation}

\noindent where \(\ket{\Psi_\rho}\) and \(\ket{\Psi_\sigma}\) range over all
purifications of \(\rho\) and \(\sigma\), respectively, in a common
extended Hilbert space.

It is also useful to note that the optimization in
Eq.~\eqref{eq:uhlmann_theorem} may be performed by fixing one
purification and optimizing only over the other \cite{Jozsa1994}.
Furthermore, it is important to emphasize that the Uhlmann fidelity itself is not a
distance. Rather, it is a similarity measure, i.e., it is symmetric, takes
values in the interval \([0,1]\), and satisfies \(F(\rho,\sigma)=1\)
when \(\rho=\sigma\), but it does not define a metric on the state
space. Nevertheless, several standard distances can be constructed as
functions of the fidelity. For example, the Bures
distance is defined as~\cite{Bengtsson2017}

\begin{equation*}
D_B(\rho,\sigma)
=
\sqrt{2\left(1-\sqrt{F(\rho,\sigma)}\right)},
\end{equation*}

\noindent while the Bures angle is given by~\cite{Nielsen2000,Bengtsson2017}

\begin{equation*}
D_{BA}(\rho,\sigma)
=
\arccos\sqrt{F(\rho,\sigma)}
\end{equation*}

Another fidelity-based distinguishability measure is the root
infidelity~\cite{Nielsen2000,Bengtsson2017},

\begin{equation*}
D_{\rm RI}(\rho,\sigma)
=
\sqrt{1-F(\rho,\sigma)}.
\end{equation*}

The entropic metric considered in this work is also a function
of the same maximal purification overlap and is defined as~\cite{Lamberti2009,Osan2013}

\begin{equation*}
D_N(\rho,\sigma)
=
\sqrt{
h_2\!\left(
\frac{1+\sqrt{F(\rho,\sigma)}}{2}
\right)
}.
\end{equation*}

Thus, different fidelity-based distances provide different scalar
quantifications of distinguishability.

\section{Purification-based entropic metric}
\label{sec:DN_metric}

For two pure states \(\ket{\psi}\) and \(\ket{\phi}\), the quantity introduced in Ref.~\cite{Lamberti2009} is

\begin{equation*}
D_N(\ket{\psi},\ket{\phi})
=
\sqrt{
H_N\!\left(
\frac{\ket{\psi}\!\bra{\psi}+\ket{\phi}\!\bra{\phi}}{2}
\right)
},
\end{equation*}

\noindent where \(H_N(\rho)=-\Tr(\rho\log_2\rho)\) is the von Neumann entropy.

If

\begin{equation*}
g \doteq |\braket{\psi}{\phi}| \in [0,1],
\end{equation*}

\noindent then

\begin{equation*}
D_N(\ket{\psi},\ket{\phi})
=
\sqrt{
\hbin\!\left(\frac{1+g}{2}\right)
},
\end{equation*}

\noindent where

\begin{equation*}
\hbin(x)\doteq -x\log_2 x -(1-x)\log_2(1-x)
\end{equation*}

\noindent is the binary entropy.

For mixed states \(\rho\) and \(\sigma\), the metric $D_N$ is defined by

\begin{equation*}
D_N(\rho,\sigma)
=
\min_{\ket{\Psi_\rho},\,\ket{\Psi_\sigma}}
D_N(\ket{\Psi_\rho},\ket{\Psi_\sigma}),
\end{equation*}

\noindent where \(\ket{\Psi_\rho}\) and \(\ket{\Psi_\sigma}\) are purifications of \(\rho\) and \(\sigma\), respectively.

Equivalently, the mixed-state problem reduces to maximizing the purification overlap,

\begin{equation*}
g_\star(\rho,\sigma)
=
\max_{\ket{\Psi_\rho},\,\ket{\Psi_\sigma}}
|\braket{\Psi_\rho}{\Psi_\sigma}|,
\end{equation*}

\noindent and then evaluating

\begin{equation*}
D_N(\rho,\sigma)
=
\sqrt{
\hbin\!\left(\frac{1+g_\star(\rho,\sigma)}{2}\right)
}.
\end{equation*}

It is worth mentioning that, by Uhlmann's theorem~\cite{Uhlmann1976,Jozsa1994,Nielsen2000,Watrous2018}, the maximal overlap is equal to the square root of the Uhlmann fidelity \cite{Jozsa1994},

\begin{equation*}
g_\star(\rho,\sigma)=\sqrt{F(\rho,\sigma)},
\end{equation*}

\noindent so that the metric \(D_N\) can also be written as

\begin{equation}
D_N(\rho,\sigma)
=
\sqrt{
\hbin\!\left(\frac{1+\sqrt{F(\rho,\sigma)}}{2}\right)
}.
\label{eq:DN_from_fidelity}
\end{equation}

However, although Eq.~\eqref{eq:DN_from_fidelity} expresses \(D_N(\rho,\sigma)\) directly in terms of the Uhlmann fidelity, the present work is not concerned solely with the scalar value of \(D_N\). Instead, we analyze the geometric structure of the optimizing purifications that attain the maximal overlap \(g_\star(\rho,\sigma)\). As we shall see in Secs.~\ref{sec:procrustes_derivation},~\ref{sec:Procrustes_solution}, and~\ref{sec:Theta_angle}, this leads to a geometric optimization problem on \(SO(3)\), whose solution encodes additional information beyond the maximal overlap \(g_\star(\rho,\sigma)\).

The generalized quantum Jensen--Shannon divergence of an ensemble
\(\mathcal{E}=\{q_i,\rho_i\}\),

\[
D_{\mathrm{JS}}(\mathcal{E})
=
S\!\left(\sum_i q_i\rho_i\right)
-
\sum_i q_i S(\rho_i),
\]

\noindent coincides with the Holevo quantity \(\chi(\mathcal{E})\)~\cite{Holevo1973,Holevo1998,Schumacher1997,Schumacher2010,Wilde2017,Roga2010}, and therefore
upper-bounds the accessible classical mutual information obtainable from
measurements on the ensemble. In the particular binary equiprobable case
\(\mathcal{E}=\{1/2,\rho;\,1/2,\sigma\}\), one has

\[
\chi(\mathcal{E})
=
D_{\mathrm{JS}}(\rho,\sigma).
\]

Moreover, the fidelity-based bound

\[
D_{\mathrm{JS}}(\rho,\sigma)
\leq
\hbin\!\left(\frac{1+\sqrt{F(\rho,\sigma)}}{2}\right)
=
D_N^2(\rho,\sigma)
\]

\noindent shows that \(D_N^2\) provides an upper bound on the Holevo quantity for
two equiprobable quantum states~\cite{Roga2010}. This bound is tight, since equality is
attained for pairs of pure states.

In this work, we adopt the metric $D_N$ because it provides a natural and meaningful scalar quantifier within the purification-based geometric framework developed here. Furthermore, $D_N$ is a bona fide metric that satisfies several appealing properties~\cite{Lamberti2009,Osan2013}.

\section{Operator-sum representation of quantum processes}
\label{sec:kraus_decomposition}
The most general deterministic evolution of a quantum state is described by a
linear map

\[
\Phi:\mathcal{D}(\mathcal{H})\longrightarrow \mathcal{D}(\mathcal{H}'),
\]

\noindent where \(\mathcal{D}(\mathcal{H})\) denotes the set of density operators on the
Hilbert space \(\mathcal{H}\). In order to represent a physically admissible
quantum evolution, \(\Phi\) is usually required to be completely positive
and trace preserving (CPTP). Complete positivity means that

\[
\mathcal{I}_n\otimes \Phi
\]

\noindent is positive for every ancillary dimension \(n\), where \(\mathcal{I}_n\) is the
identity map acting on the ancillary system. This requirement guarantees that
\(\Phi\) remains positive even when the system is part of a larger
entangled state. Trace preservation,

\[
\operatorname{Tr}\Phi(\rho)=\operatorname{Tr}\rho,
\]

expresses conservation of total probability.

Every CPTP map admits an operator-sum, or Kraus, representation of the form~\cite{Kraus1971,Kraus1983,Nielsen2000,Schumacher2010,Bengtsson2017}

\[
\Phi(\rho)
=
\sum_{\mu} K_\mu \rho K_\mu^\dagger ,
\]

\noindent where the operators

\[
K_\mu:\mathcal{H}\longrightarrow\mathcal{H}'
\]

\noindent are called Kraus operators. The trace-preserving condition is equivalent to

\[
\sum_\mu K_\mu^\dagger K_\mu = \I_{\mathcal{H}}.
\]

Conversely, any map written in this form with Kraus operators satisfying the
above completeness relation is completely positive and trace preserving.

A particularly important subclass is given by unitary channels,

\[
\Phi(\rho)=U\rho U^\dagger,
\]

\noindent which correspond to the case of a single Kraus operator \(K_1=U\). Unitary
channels describe the reversible dynamics of closed quantum systems. More
general Kraus representations describe irreversible processes such as noise,
decoherence, dissipation, and measurement without postselection.

The connection with open quantum systems arises naturally from unitary
evolution on a larger Hilbert space. Let \(S\) denote the system of interest
and \(E\) its environment. Suppose that the initial joint state is factorized,

\[
\rho_{SE}(0)=\rho_S(0)\otimes \omega_E,
\]

\noindent and that the composite system evolves unitarily according to

\[
\rho_{SE}(t)
=
U_{SE}(t)\,
\bigl(\rho_S(0)\otimes \omega_E\bigr)\,
U_{SE}^\dagger(t).
\]

The reduced state of the system at time \(t\) is obtained by tracing out the
environment:

\[
\rho_S(t)
=
\operatorname{Tr}_E\!\left[
U_{SE}(t)\,
\bigl(\rho_S(0)\otimes \omega_E\bigr)\,
U_{SE}^\dagger(t)
\right].
\]

This defines a quantum channel

\[
\rho_S(t)=\Phi_t(\rho_S(0)).
\]

To obtain the Kraus representation explicitly, let the environment state have
spectral decomposition

\[
\omega_E=\sum_\alpha p_\alpha |\alpha\rangle\langle\alpha|,
\]

\noindent and let \(\{|e_\beta\rangle\}_\beta\) be an orthonormal basis of the
environment. Then

\[
\Phi_t(\rho)
=
\sum_{\alpha,\beta}
K_{\beta\alpha}(t)\,\rho\,K_{\beta\alpha}^\dagger(t),
\]

\noindent with Kraus operators

\[
K_{\beta\alpha}(t)
=
\sqrt{p_\alpha}\,
\langle e_\beta|U_{SE}(t)|\alpha\rangle.
\]

Here the matrix element is taken only on the environmental Hilbert space, so
each \(K_{\beta\alpha}(t)\) acts on the system Hilbert space. The unitarity of
\(U_{SE}(t)\) implies

\[
\sum_{\alpha,\beta}
K_{\beta\alpha}^\dagger(t)K_{\beta\alpha}(t)
=
\I_S,
\]

\noindent and therefore \(\Phi_t\) is CPTP.

Thus, the Kraus representation provides the standard operational description
of open-system dynamics when the system and environment are initially
uncorrelated. The apparent irreversibility of the reduced dynamics arises from
discarding environmental degrees of freedom through the partial trace, even
though the joint system--environment evolution is unitary.

The Kraus representation provides the standard operational description
of open-system dynamics when the system and environment are initially
uncorrelated. The apparent irreversibility of the reduced dynamics arises from
discarding environmental degrees of freedom through the partial trace, even
though the joint system--environment evolution is unitary.

\section{Quantum channels}
\label{sec:quantum_channels}

\subsection{Pauli-type channels}
\label{sec:pauli}

\subsubsection{Depolarizing channel}
\label{sssec:depolarizing_channel}

The depolarizing channel describes isotropic noise that drives a qubit
toward the maximally mixed state. In the convention used here, the
channel is defined as~\cite{Nielsen2000}

\begin{equation*}
\Phi_p^{\mathrm{dep}}(\rho)
=
(1-p)\rho
+
p\,\frac{\I}{2},
\qquad
0\leq p\leq 1 .
\end{equation*}

Let the input state be written in Bloch form as~\cite{Nielsen2000}

\begin{equation*}
\rho(\mathbf r)
=
\frac{1}{2}
\left(
\I+\mathbf r\cdot\boldsymbol{\sigma}
\right),
\qquad
\mathbf r=(r_x,r_y,r_z).
\end{equation*}

The action on the Bloch vector $\mathbf r$ is simply

\begin{equation*}
\mathbf r
\longmapsto
\mathbf r'
=
(1-p)\mathbf r .
\end{equation*}

\subsubsection{Bit-flip channel}
\label{sssec:bit_flip_channel}

The bit-flip channel describes a stochastic error in which a Pauli
\(\sigma_x\) operation is applied to the qubit with probability \(p\).
It is defined by~\cite{Nielsen2000}

\begin{equation*}
\Phi_p^{\mathrm{BF}}(\rho)
=
(1-p)\rho+p\,\sigma_x\rho\sigma_x ,
\qquad 0\leq p\leq 1 .
\end{equation*}

Equivalently, it admits the Kraus representation

\begin{equation*}
K_0=\sqrt{1-p}\,\I,
\qquad
K_1=\sqrt{p}\,\sigma_x ,
\end{equation*}

\noindent which satisfies

\begin{equation*}
K_0^\dagger K_0+K_1^\dagger K_1=\I .
\end{equation*}

Let the input state be written as

\begin{equation*}
\rho(\mathbf r)
=
\frac{1}{2}
\left(
\I+\mathbf r\cdot\boldsymbol{\sigma}
\right),
\qquad
\mathbf r=(r_x,r_y,r_z).
\end{equation*}

\noindent Using

\begin{equation*}
\sigma_x\sigma_x\sigma_x=\sigma_x,
\qquad
\sigma_x\sigma_y\sigma_x=-\sigma_y,
\qquad
\sigma_x\sigma_z\sigma_x=-\sigma_z,
\end{equation*}

\noindent one obtains

\begin{equation*}
\Phi_p^{\mathrm{BF}}(\rho(\mathbf r))
=
\frac{1}{2}
\left[
\I
+
r_x\sigma_x
+
(1-2p)r_y\sigma_y
+
(1-2p)r_z\sigma_z
\right].
\end{equation*}

Thus, in the Bloch representation, the channel acts as

\begin{equation*}
\mathbf r=(r_x,r_y,r_z)
\longmapsto
\mathbf r'
=
\bigl(r_x,\lambda r_y,\lambda r_z\bigr),
\qquad
\lambda=1-2p .
\end{equation*}

\subsubsection{Phase-flip channel}
\label{sssec:phase_flip_channel}

The phase-flip channel describes a stochastic phase error in which a
Pauli \(\sigma_z\) operation is applied to the qubit with probability
\(p\). It is defined by~\cite{Nielsen2000}

\begin{equation*}
\Phi_p^{\mathrm{PF}}(\rho)
=
(1-p)\rho
+
p\,\sigma_z\rho\sigma_z,
\qquad
0\leq p\leq 1 .
\end{equation*}

A Kraus representation is given by

\begin{equation*}
K_0=\sqrt{1-p}\,\I,
\qquad
K_1=\sqrt{p}\,\sigma_z,
\end{equation*}

which satisfies

\begin{equation*}
K_0^\dagger K_0+K_1^\dagger K_1=\I .
\end{equation*}

Let the input state be written in Bloch form as

\begin{equation*}
\rho(\mathbf r)
=
\frac{1}{2}
\left(
\I+\mathbf r\cdot\boldsymbol{\sigma}
\right),
\qquad
\mathbf r=(r_x,r_y,r_z).
\end{equation*}

Using

\begin{equation*}
\sigma_z\sigma_x\sigma_z=-\sigma_x,
\qquad
\sigma_z\sigma_y\sigma_z=-\sigma_y,
\qquad
\sigma_z\sigma_z\sigma_z=\sigma_z,
\end{equation*}

\noindent one obtains

\begin{equation*}
\Phi_p^{\mathrm{PF}}(\rho(\mathbf r))
=
\frac{1}{2}
\left[
\I
+
(1-2p)r_x\sigma_x
+
(1-2p)r_y\sigma_y
+
r_z\sigma_z
\right].
\end{equation*}

Thus, in the Bloch representation, the phase-flip channel acts as

\begin{equation*}
\mathbf r=(r_x,r_y,r_z)
\longmapsto
\mathbf r'
=
\bigl(\lambda r_x,\lambda r_y,r_z\bigr),
\qquad
\lambda=1-2p .
\end{equation*}

\subsection{Amplitude damping channel}
\label{sssec:amplitude_damping_channel}

The amplitude damping channel $\mathcal{A}_\gamma(\rho)$ describes dissipative relaxation of a qubit
toward its ground state. It is a standard model for energy loss, for
instance spontaneous emission of a two-level atom or relaxation in a
spin system. The channel is defined by the Kraus operators~\cite{Nielsen2000}

\begin{align*}
K_0 &=
\ket{0}\!\bra{0}
+
\sqrt{1-\gamma}\,\ket{1}\!\bra{1},
\\
K_1 &=
\sqrt{\gamma}\,\ket{0}\!\bra{1},
\end{align*}

\noindent where \(0\leq \gamma \leq 1\) is the damping probability. These operators
satisfy

\begin{equation*}
K_0^\dagger K_0+K_1^\dagger K_1=\I,
\end{equation*}

\noindent and the channel action is

\begin{equation*}
\mathcal{A}_\gamma(\rho)
=
K_0\rho K_0^\dagger
+
K_1\rho K_1^\dagger .
\end{equation*}

Let the input state be written in Bloch form as

\begin{equation*}
\rho(\mathbf r)
=
\frac{1}{2}
\left(
\I+\mathbf r\cdot\boldsymbol{\sigma}
\right),
\qquad
\mathbf r=(r_x,r_y,r_z).
\end{equation*}

\noindent A direct calculation gives

\begin{align*}
\mathcal{A}_\gamma(\rho(\mathbf r))
=&
\frac{1}{2}
\Big[
\I
+
\sqrt{1-\gamma}\,r_x\sigma_x
+
\sqrt{1-\gamma}\,r_y\sigma_y\\
&+
\bigl(\gamma+(1-\gamma)r_z\bigr)\sigma_z
\Big].
\end{align*}

Thus, in the Bloch representation, the amplitude damping channel $\mathcal{A}_\gamma(\rho)$ acts as
the affine transformation 

\[
\mathbf r=(r_x,r_y,r_z) \longmapsto \mathbf r'
\]

\noindent where

\begin{equation*}
\mathbf r'
=
\left(
\sqrt{1-\gamma}\,r_x,\,
\sqrt{1-\gamma}\,r_y,\,
\gamma+(1-\gamma)r_z
\right).
\end{equation*}

\subsection{Imperfect quantum NOT gate}
\label{ssec:imperfect_not}

The quantum NOT gate for a single qubit is implemented by the unitary
transformation~\cite{Nielsen2000}

\begin{equation}
\mathcal{N}(\rho) = \sigma_x \rho \sigma_x^\dagger,
\label{eq:ideal_not_bloch}
\end{equation}

\noindent which corresponds, up to a global phase, to a rotation of angle $\pi$
about the $x$-axis. In the Bloch representation, this operation maps

\begin{equation*}
\mathbf r = (x,y,z) \;\longmapsto\; \mathbf r' = (x,-y,-z),
\end{equation*}

\noindent and therefore inverts the Bloch vector with respect to the $x$-axis.

In realistic implementations, however, control errors may lead to
imperfect rotations. A common source of such imperfections is a
miscalibrated control pulse, resulting in a rotation angle that deviates
from $\pi$. To model this situation, we consider a stochastic mixture of
unitary rotations about the $x$-axis.

Specifically, we assume that with probability $(1-p)$ the ideal NOT gate
is applied, while with probability $p$ the system undergoes a slightly
imperfect rotation of angle $\pi+\delta\alpha$. Defining the unitary transformation

\begin{equation*}
U_\delta = e^{-i \frac{(\pi+\delta\alpha)}{2}\sigma_x},
\end{equation*}

\noindent the corresponding quantum channel is given by

\begin{equation}
\Phi^{\mathrm{NOT}}_{p}(\rho)
=
(1-p)\,\sigma_x \rho \sigma_x^\dagger
+
p\, U_\delta \rho U_\delta^\dagger.
\label{eq:imperfect_not_channel}
\end{equation}

This defines a completely positive trace-preserving map with Kraus
operators

\begin{equation*}
K_0 = \sqrt{1-p}\,\sigma_x,
\qquad
K_1 = \sqrt{p}\,U_\delta,
\end{equation*}

\noindent which satisfy $\sum_i K_i^\dagger K_i = \I$.\par

Let $\rho(\mathbf{r}) = \frac{1}{2}(\I + \mathbf{r}\cdot\boldsymbol{\sigma})$
be the Bloch representation of the input state. Since both $\sigma_x$
and $U_\delta$ correspond to rotations about the $x$-axis, the channel
acts as a probabilistic mixture of such rotations. A direct calculation shows that
\(\Phi^{\mathrm{NOT}}_{p}(\rho)\) maps
\(\mathbf r = (r_x,r_y,r_z)\) to \(\mathbf r'\), where

\begin{equation*}
\mathbf{r}' =
\begin{pmatrix}
r_x \\
-\left[(1-p) + p\cos(\delta\alpha)\right] r_y
+ p \sin(\delta\alpha)\, r_z \\
-\left[(1-p) + p\cos(\delta\alpha)\right] r_z
- p \sin(\delta\alpha)\, r_y
\end{pmatrix}.
\end{equation*}

Thus, the imperfect NOT gate $\Phi^{\mathrm{NOT}}_{p}(\rho)$ leaves \(r_x\) invariant, while the
$(y,z)$ components undergo a rotation and attenuation determined by the
parameters $p$ and $\delta\alpha$. In the limit $\delta\alpha=0$, one
recovers the ideal NOT gate \eqref{eq:ideal_not_bloch}, while for small $\delta\alpha$ the channel
describes a weak coherent error around the $x$-axis.\par

\section{Qubit purifications in Fano form and the Procrustes reduction}
\label{sec:fano_procrustes}

In this section we summarize the algebraic structure underlying the
Fano--Procrustes construction used in this work.
Our aim is to show explicitly how the optimization over purifications of
mixed one-qubit states reduces to an orthogonal Procrustes problem \cite{Schonemann1966}.\par

The parametrization used in this section is based on the characterization of qubit purifications presented in~\cite{Constantinescu2002}, where it is shown that all purifications of a mixed qubit state can be generated via right multiplication of a fixed solution by elements of the Lie group $SO(3,\mathbb{R}^3)$. We adopt the corresponding Fano-form representation and adapt it to the optimization problem considered here.

\subsection{Fano representation of a two-qubit purification}

Let

\begin{equation*}
\rho(\mathbf r)=\frac{1}{2}\left(\I+\mathbf r\cdot\boldsymbol{\sigma}\right),
\qquad
\|\mathbf r\|<1,
\end{equation*}

\noindent be a mixed qubit state. A purification of \(\rho(\mathbf r)\) may be represented by the pure two-qubit projector

\begin{equation*}
P_\rho=\ket{\Psi_\rho}\bra{\Psi_\rho},
\end{equation*}

\noindent which admits the expansion~\cite{Constantinescu2002}

\begin{equation*}
P_\rho
=
\frac{1}{4}
\left(
\I\otimes \I
+\mathbf r\cdot\boldsymbol{\sigma}\otimes \I
+\boldsymbol{\gamma}\cdot \I\otimes\boldsymbol{\sigma}
+\sum_{i,j=1}^{3}A_{ij}\,\sigma_i\otimes\sigma_j
\right),
\end{equation*}

\noindent where \(\boldsymbol{\gamma}=A^{T}\mathbf r\in\mathbb{R}^{3}\) denotes the ancilla Bloch vector, and \(A\in\mathbb{R}^{3\times 3}\) is the correlation matrix, whose entries are given by

\[
A_{ij}=\Tr(P_\rho\,\sigma_i\otimes\sigma_j).
\]

Tracing over the ancilla gives

\begin{equation*}
\Tr_{2}(P_\rho)=\rho(\mathbf r),
\end{equation*}

\noindent while tracing over the system gives the ancillary reduced state

\begin{equation*}
\Tr_{1}(P_\rho)=\frac{1}{2}\left(\I+\boldsymbol{\gamma}\cdot\boldsymbol{\sigma}\right).
\end{equation*}

\subsection{Purity constraints}

The coefficients \(A_{ij}\) of the matrix $A$ represent the correlations between Pauli observables on the system and
the ancilla, and therefore define a correlation matrix. However, for
pure states (e.g., purifications), these coefficients are not independent,
but are constrained by the purity condition 

\begin{equation}
P_\rho^2=P_\rho.
\label{eq:app_idempotent}
\end{equation}

By expanding Eq.~\eqref{eq:app_idempotent} in the Pauli basis, one obtains
the standard purity constraints for two-qubit pure states. The most relevant
ones for our purposes are

\begin{align}
AA^T &= (1-r^2)\I_3+\mathbf r\,\mathbf r^T,
\label{eq:app_AAT}
\\
A^T A &= (1-\gamma^2)\I_3+\boldsymbol{\gamma}\,\boldsymbol{\gamma}^T,
\label{eq:app_ATA}
\end{align}

\noindent together with

\begin{equation*}
(\det A)^2=(1-\|\mathbf r\|^2)^2=(1-\|\boldsymbol{\gamma}\|^2)^2.
\end{equation*}

In the convention used throughout this work, we choose

\begin{equation}
\det A = \|\mathbf r\|^2-1 = \|\boldsymbol{\gamma}\|^2-1.
\label{eq:app_detA}
\end{equation}

In particular, Eqs.~\eqref{eq:app_AAT}--\eqref{eq:app_detA} imply

\begin{equation*}
\|\boldsymbol{\gamma}\|=\|\mathbf r\|.
\end{equation*}

\subsection{Overlap of purifications in Fano variables}
\label{sec:overlap_fano}

Let \(P_\rho\) and \(P_\sigma\) be purifications of the states

\begin{equation*}
\rho=\frac{1}{2}(\I+\mathbf r\cdot\boldsymbol{\sigma}),
\qquad
\sigma=\frac{1}{2}(\I+\mathbf s\cdot\boldsymbol{\sigma}),
\end{equation*}

\noindent with Fano parameters \((\mathbf r,\boldsymbol{\gamma},A)\) and
\((\mathbf s,\boldsymbol{\delta},B)\), respectively.
Since \(P_\rho\) and \(P_\sigma\) are rank-one projectors,

\begin{equation*}
\Tr(P_\rho P_\sigma)=|\braket{\Psi_\rho}{\Psi_\sigma}|^2.
\end{equation*}

Using the orthogonality relations

\begin{equation*}
\Tr(\sigma_i\sigma_j)=2\delta_{ij},
\qquad
\Tr(\sigma_i)=0,
\end{equation*}

\noindent one obtains

\begin{equation}
|\braket{\Psi_\rho}{\Psi_\sigma}|^2
=
\frac{1}{4}
\left[
1+\mathbf r\cdot\mathbf s
+\boldsymbol{\gamma}\cdot\boldsymbol{\delta}
+\Tr(A^T B)
\right].
\label{eq:overlap_fano}
\end{equation}

\subsection{General form of the purification variables}

Let \((\widetilde{\boldsymbol{\gamma}},\widetilde A)\) denote a particular
choice of ancilla Bloch vector and correlation matrix satisfying
Eqs.~\eqref{eq:app_AAT}--\eqref{eq:app_detA} for the state \(\rho(\mathbf r)\).
Then every other purification of the same state can be written as

\begin{equation*}
A(S)=\widetilde A\,S,
\qquad
\boldsymbol{\gamma}(S)=S^T\widetilde{\boldsymbol{\gamma}},
\qquad
S\in SO(3).
\end{equation*}

Indeed, if \(A\) and \(\widetilde A\) satisfy

\begin{equation*}
AA^T=\widetilde A \widetilde A^T,
\end{equation*}

\noindent then, since we restrict attention to mixed states, both matrices are
invertible and one may define

\begin{equation*}
S=\widetilde A^{-1}A.
\end{equation*}

It follows immediately that \(S\in O(3)\). Moreover, if both \(A\) and
\(\widetilde A\) are chosen with the convention \eqref{eq:app_detA}, then

\begin{equation*}
\det S=\frac{\det A}{\det\widetilde A}=1,
\end{equation*}

\noindent hence \(S\in SO(3)\).

The corresponding transformation law for the ancilla Bloch vector is
fixed by the relation \(\boldsymbol{\gamma}=A^T\mathbf r\). Indeed,

\begin{equation*}
\boldsymbol{\gamma}(S)
=
A(S)^T\mathbf r
=
S^T\widetilde A^T\mathbf r
=
S^T\widetilde{\boldsymbol{\gamma}}.
\end{equation*}

This is consistent with Eq.~\eqref{eq:app_ATA}.

\section{Derivation of the orthogonal Procrustes problem}
\label{sec:procrustes_derivation}

Fix a purification of \(\rho\), represented by Fano parameters
\((\mathbf r,\boldsymbol{\gamma},A)\), and choose a particular purification
of \(\sigma\), represented by
\((\mathbf s,\widetilde{\boldsymbol{\delta}},\widetilde B)\).
All purifications of \(\sigma\) are then represented as

\begin{equation}
B(S)=\widetilde B\,S,
\qquad
\boldsymbol{\delta}(S)=S^T\widetilde{\boldsymbol{\delta}},
\qquad
S\in SO(3).
\label{eq:B_of_S}
\end{equation}

Substituting Eq.~\eqref{eq:B_of_S} into Eq.~\eqref{eq:overlap_fano}, we obtain

\begin{equation*}
g^2(S)
=
\frac{1}{4}
\left[
1+\mathbf r\cdot\mathbf s
+\boldsymbol{\gamma}\cdot S^T\widetilde{\boldsymbol{\delta}}
+\Tr\!\left(A^T\widetilde B\,S\right)
\right].
\end{equation*}

Using the identity

\begin{equation*}
\boldsymbol{\gamma}\cdot S^T\widetilde{\boldsymbol{\delta}}
=
\Tr\!\left(
\boldsymbol{\gamma}\,\widetilde{\boldsymbol{\delta}}^{\,T} S
\right),
\end{equation*}

\noindent the \(S\)-dependent part can be written as a single trace:

\begin{equation*}
g^2(S)
=
\frac{1}{4}
\left[
1+\mathbf r\cdot\mathbf s
+\Tr(KS)
\right],
\end{equation*}

\noindent where

\begin{equation*}
K
\doteq
A^T\widetilde B
+
\boldsymbol{\gamma}\,\widetilde{\boldsymbol{\delta}}^{\,T}.
\end{equation*}

Therefore maximizing the purification overlap is equivalent to

\begin{equation}
\max_{S\in SO(3)} \Tr(KS).
\label{eq:procrustes_trace}
\end{equation}

This is the appropriate orthogonal Procrustes problem in the present setting~\cite{Schonemann1966}.

It is also useful to introduce the augmented matrices

\begin{equation*}
\widehat A \doteq
\begin{pmatrix}
A\\
\boldsymbol{\gamma}^{\,T}
\end{pmatrix},
\qquad
\widehat B \doteq
\begin{pmatrix}
\widetilde B\\
\widetilde{\boldsymbol{\delta}}^{\,T}
\end{pmatrix},
\end{equation*}

\noindent for which

\begin{equation*}
\widehat A^{\,T}\widehat B
=
A^T\widetilde B+\boldsymbol{\gamma}\,\widetilde{\boldsymbol{\delta}}^{\,T}
=
K.
\end{equation*}

Then

\begin{align*}
\|\widehat B S-\widehat A\|_{\HS}^2
&=
\Tr\!\left[(\widehat B S-\widehat A)^T(\widehat B S-\widehat A)\right]
\nonumber\\
&=
\Tr(\widehat B^T\widehat B)+\Tr(\widehat A^T\widehat A)-2\Tr(KS),
\end{align*}

\noindent and hence

\begin{equation*}
\max_{S\in SO(3)}\Tr(KS)
\iff
\min_{S\in SO(3)}\|\widehat B S-\widehat A\|_{\HS}.
\end{equation*}

Let

\begin{equation*}
K=U\Sigma V^T
\end{equation*}

\noindent be a singular value decomposition of \(K\)~\cite{Horn1986}, with

\noindent \(\Sigma=\diag(\sigma_1,\sigma_2,\sigma_3)\), \(\sigma_i\ge 0\).

Then the maximizing proper rotation is

\begin{equation*}
S_\star = V \Lambda U^T,
\qquad
\Lambda=\diag(1,1,\det(VU^T)).
\end{equation*}

The optimal purification of \(\sigma\) is therefore

\begin{equation*}
B_\star=\widetilde B\,S_\star,
\qquad
\boldsymbol{\delta}_\star=S_\star^T\widetilde{\boldsymbol{\delta}},
\end{equation*}

\noindent and the maximal overlap becomes

\begin{equation}
g_\star^2
=
\frac{1}{4}
\left[
1+\mathbf r\cdot\mathbf s
+\boldsymbol{\gamma}\cdot\boldsymbol{\delta}_\star
+\Tr(A^T B_\star)
\right].
\label{eq:gstar_corrected}
\end{equation}

\section{Optimal Procrustes rotation for qubit purifications}
\label{sec:Procrustes_solution}

In this section we present the explicit form of the orthogonal
Procrustes problem that determines the optimal rotation
\(S_\star\in SO(3)\) relating purifications of two mixed qubit states
\(\rho\) and \(\sigma\).

\subsection{A suitable canonical gauge}

Let

\begin{equation*}
r = \|\mathbf r\|, \qquad \mathbf n_r = \frac{\mathbf r}{r} \quad (r \neq 0),
\end{equation*}

\noindent and let \(O(\mathbf n_r) \in SO(3)\) satisfy

\begin{equation*}
O(\mathbf n_r)\,\hat{\mathbf z} = \mathbf n_r.
\end{equation*}

Define

\begin{equation*}
\alpha_r = \sqrt{1 - r^2}, 
\qquad 
D_r = \operatorname{diag}(\alpha_r, \alpha_r, -1).
\end{equation*}

Then

\begin{equation*}
A_r = O(\mathbf n_r)\, D_r
\end{equation*}

\noindent satisfies Eq.~\eqref{eq:app_AAT}, since

\begin{align*}
A_r A_r^T 
&= O(\mathbf n_r) D_r^2 O(\mathbf n_r)^T\\
&= O(\mathbf n_r)\,\operatorname{diag}(1 - r^2, 1 - r^2, 1)\,O(\mathbf n_r)^T\\
&= (1 - r^2) \I_3 + \mathbf r \mathbf r^T.
\end{align*}

In this gauge, the ancilla Bloch vector is determined by

\begin{equation*}
\boldsymbol{\gamma}_r = A_r^T \mathbf r.
\end{equation*}

Using \(\mathbf r = r\, O(\mathbf n_r)\hat{\mathbf z}\), one obtains

\begin{equation*}
\boldsymbol{\gamma}_r
= D_r O(\mathbf n_r)^T \mathbf r
= r\, D_r \hat{\mathbf z}
= -\,r\, \hat{\mathbf z}.
\end{equation*}

\noindent With this choice one has

\begin{equation*}
A_r^T A_r
= \operatorname{diag}(1 - r^2, 1 - r^2, 1)
= (1 - r^2) \I_3 + \boldsymbol{\gamma}_r \boldsymbol{\gamma}_r^T,
\end{equation*}

\noindent in agreement with Eq.~\eqref{eq:app_ATA}.

The same construction applies to the state \(\sigma\), yielding

\begin{equation*}
A_s = O(\mathbf n_s) D_s, 
\qquad 
\widetilde{\boldsymbol{\delta}}_s = -\,s\, \hat{\mathbf z},
\end{equation*}

\noindent with

\begin{align*}
s =& \|\mathbf s\|, 
\qquad 
\mathbf n_s = \frac{\mathbf s}{s},\\
D_s =& \operatorname{diag}(\alpha_s, \alpha_s, -1),
\qquad 
\alpha_s = \sqrt{1 - s^2}.
\end{align*}

\begin{remark}
The rotation \(O(\mathbf n)\) is not unique, since any additional
right multiplication by a rotation around \(\hat{\mathbf z}\) leaves
\(O(\mathbf n)\hat{\mathbf z}=\mathbf n\) unchanged. Throughout this
work we assume a fixed convention for choosing \(O(\mathbf n)\), for
example the minimal rotation that maps \(\hat{\mathbf z}\) into
\(\mathbf n\), with a prescribed choice in the exceptional case
\(\mathbf n=-\hat{\mathbf z}\). Once this convention is fixed, the
canonical gauge is well defined.
\end{remark}

\begin{remark}
The sign of the ancilla Bloch vector is fixed by the relation
\(\boldsymbol{\gamma}_r = A_r^T \mathbf r\).
While the choice \(\boldsymbol{\gamma}_r = \pm r\,\hat{\mathbf z}\)
is equivalent at the level of quadratic expressions such as
\(\boldsymbol{\gamma}_r \boldsymbol{\gamma}_r^T\),
the negative sign arises naturally from the canonical choice
\(D_r = \operatorname{diag}(\alpha_r,\alpha_r,-1)\).
\end{remark}

\subsection{General formulation}

In the canonical gauge, the Procrustes matrix is

\begin{equation*}
K
=
A_r^T A_s
+
\boldsymbol{\gamma}_r\,\widetilde{\boldsymbol{\delta}}_s^{\,T}.
\end{equation*}

Since

\begin{equation*}
A_r^T A_s
=
D_r\,O(\mathbf n_r)^T O(\mathbf n_s)\,D_s,
\end{equation*}

\noindent we may write

\begin{equation}
K
=
D_r\,R_{rs}\,D_s
+
rs\,\hat{\mathbf z}\hat{\mathbf z}^{\,T},
\label{eq:K_canonical}
\end{equation}

\noindent where

\begin{equation*}
R_{rs}\doteq O(\mathbf n_r)^T O(\mathbf n_s)\in SO(3)
\end{equation*}

\noindent is the relative rotation between the canonical frames associated with
\(\mathbf n_r\) and \(\mathbf n_s\).

The optimal proper rotation is then obtained from the singular value
decomposition

\begin{equation*}
K=U\Sigma V^T
\end{equation*}

\noindent through

\begin{equation*}
S_\star = V\Lambda U^T,
\qquad
\Lambda=\diag(1,1,\det(VU^T)).
\end{equation*}

Once \(S_\star\) is known, the optimal purification data of \(\sigma\) are

\begin{equation*}
B_\star = A_s S_\star,
\qquad
\boldsymbol{\delta}_\star=S_\star^T\widetilde{\boldsymbol{\delta}}_s,
\end{equation*}

and the maximal overlap is given by Eq.~\eqref{eq:gstar_corrected}.

\section{Procrustes optimal rotation angle}
\label{sec:Theta_angle}

The Procrustes formulation naturally yields a geometric quantity
associated with the alignment of purification frames. Any rotation \(S\in SO(3)\) admits an axis--angle representation

\begin{equation*}
S=\exp\!\left(\theta\,\hat{\mathbf u}\cdot \mathbf J\right),
\end{equation*}

\noindent where \(\theta\in[0,\pi]\), \(\hat{\mathbf u}\in\mathbb R^3\) is a unit vector,
and \(\mathbf J=(J_x,J_y,J_z)\) are the generators of \(\mathfrak{so}(3)\).
The rotation angle can be extracted from the trace via

\begin{equation*}
\theta_S
=
\arccos\!\left(\frac{\Tr(S)-1}{2}\right).
\end{equation*}

Let \(S_\star\in SO(3)\) be a solution of the Procrustes problem
\eqref{eq:procrustes_trace} in the canonical gauge. We define the \emph{purification misalignment angle}
as

\begin{equation*}
\Theta(\rho,\sigma)
=
\arccos\!\left(
\frac{\Tr(S_\star)-1}{2}
\right).
\end{equation*}

It is straightforward to verify that
\(\Theta(\rho,\rho)=0\), since in this case the Procrustes matrix is diagonal and the identity element \(\I_3\) of \(SO(3)\) achieves the maximal overlap.

For channel dynamics \(\rho\mapsto\Phi(\rho)\), the quantity
\(\Theta(\rho,\Phi(\rho))\) measures the rotation required to align the
optimal purification frames associated with the input and output states, within
the chosen gauge.

In contrast with \(D_N\), which depends only on the optimized overlap,
\(\Theta\) retains geometric information about how the optimal purifications
must be reoriented to achieve that overlap. In particular, if the Bloch vectors are parallel (i.e., \(\mathbf n_s=\mathbf n_r\)), then

\begin{equation*}
R_{rs}=\I_3,
\end{equation*}

\noindent and Eq.~\eqref{eq:K_canonical} reduces to

\begin{equation*}
K
=
D_r D_s + rs\,\hat{\mathbf z}\hat{\mathbf z}^{\,T}
=
\diag(\alpha_r\alpha_s,\alpha_r\alpha_s,1+rs).
\end{equation*}

Since \(K\) is diagonal and positive semidefinite, one finds

\begin{equation*}
S_\star=\I_3,
\end{equation*}

\noindent and therefore

\begin{equation*}
\Theta(\rho,\sigma)=0.
\end{equation*}

Thus, within the canonical purification gauge, \(\Theta\) provides a complementary
geometric indicator of the alignment structure underlying the Procrustes
optimization, supplementing the scalar information captured by the metric
\(D_N\).

\section{Purification misalignment and channel anisotropy}
\label{sec:theta_anisotropy}

Any completely positive trace-preserving (CPTP) map acting on a qubit can be
written in affine Bloch form as

\begin{equation*}
\mathbf s = M\mathbf r + \mathbf c,
\end{equation*}

\noindent where \(\mathbf r\) and \(\mathbf s\) are the Bloch vectors of the input and
output states, \(M\) is a real \(3\times 3\) matrix describing the linear
part of the transformation, and \(\mathbf c\in\mathbb R^3\) is a translation
vector.
The term \(M\mathbf r\) accounts for contraction and anisotropic deformation
of the Bloch sphere, while \(\mathbf c\) represents the affine shift
characteristic of nonunital channels.

A particularly simple situation occurs when the channel preserves the Bloch
direction, so that the input and output vectors are parallel:

\begin{equation*}
\mathbf s = \mu(\mathbf r)\,\mathbf r,
\qquad
\mu(\mathbf r)\ge 0.
\end{equation*}

In this case the Procrustes matrix \(K\) is diagonal in the
canonical gauge, and one finds

\begin{equation*}
S_\star=\I_3,
\qquad
\Theta(\rho,\Phi(\rho))=0.
\end{equation*}

Thus purely radial motion inside the Bloch ball changes the mixedness of the
state but does not twist the optimal purification frame.

When the channel does not preserve the Bloch direction,
the output Bloch direction differs from the input one and the 
Procrustes matrix \(K\) is no longer diagonal in the canonical gauge.
The optimal alignment is then generally nontrivial, and one obtains

\begin{equation*}
\Theta(\rho,\Phi(\rho))>0
\end{equation*}

\noindent for generic inputs.

The purification misalignment angle therefore isolates the component of the
channel action that changes the orientation of the optimal purification frame.
Radial changes of the Bloch radius leave \(\Theta\) equal to zero, while
anisotropic deformations and affine drifts may produce a finite value.

In this sense \(\Theta(\rho,\Phi(\rho))\) may be viewed as a geometric
indicator of the nonradial part of the channel action relative to the chosen
input state.
While the scalar metric $D_N$ measures the magnitude of the optimized
purification overlap, the angle \(\Theta\) retains directional information
about how the channel twists the corresponding optimal purification frame.

\section{Analytical results}
\label{sec:analytical_results}

In this section we provide closed-form expressions for the purification
overlap and the associated metric $D_N$ for several important classes of
quantum channels. The analysis focuses on families of input states that are
adapted to the symmetry of the channel, so that the Bloch vectors of the
input and output states remain collinear. In such cases, the Procrustes
 problem becomes trivial, the optimal rotation reduces to the
identity, and the overlap can be computed analytically without performing
an explicit singular-value optimization. This allows for simple expressions
of both the overlap and the metric in terms of the Bloch parameters.

\subsection{Pauli-type channels}
\label{sec:pauli_analytical}

\subsubsection{Depolarizing channel}

The depolarizing channel $\Phi_p^{\mathrm{dep}}(\rho)$, with error probability \(p\), acts on the Bloch vector as (cf.~Sec.\ref{sssec:depolarizing_channel})

\begin{equation*}
\mathbf r
\longmapsto
\mathbf r'
=
(1-p)\mathbf r .
\end{equation*}

The channel is thus purely radial: it preserves the direction of the
Bloch vector and decreases only its length. For \(p=0\) the channel is
the identity map, while for \(p=1\) every input state is mapped to the
maximally mixed state.

The depolarizing channel provides an important reference case for the
Fano--Procrustes framework. Since the input and output Bloch vectors are
parallel for all \(p\in[0,1]\), the canonical Procrustes matrix is
diagonal in the adapted frame and the optimal rotation is trivial:

\begin{equation*}
S_\star=\I_3 .
\end{equation*}

Consequently,

\begin{equation*}
\Theta\!\left(\rho,\Phi_p^{\mathrm{dep}}(\rho)\right)=0
\qquad
\text{for all } p\in[0,1].
\end{equation*}

Nevertheless, the metric $D_N$ is generally nonzero for \(p>0\),
because the channel changes the purity of the state by contracting the
Bloch vector. Therefore (cf. Appendix~\ref{app:collinear})

\begin{eqnarray*}
g_{\mathrm{dep}}(r,p)
=&
\frac{1}{2}
\left[
\sqrt{(1+r)(1+(1-p)r)}\right.\\
&+
\left.\sqrt{(1-r)(1-(1-p)r)}
\right],
\end{eqnarray*}

\begin{equation*}
D_N\!\left(\rho,\Phi_p^{\mathrm{dep}}(\rho)\right)
=
\sqrt{
\hbin\!\left(
\frac{1+g_{\mathrm{dep}}(r,p)}{2}
\right)
}.
\end{equation*}

\subsubsection{Bit-flip channel}

The bit-flip channel $\Phi^{BF}_p(\rho)$, with error probability \(p\), acts on the Bloch vector as (cf.~Sec.~\ref{sssec:bit_flip_channel})

\begin{equation*}
\mathbf r = (x,y,z)\mapsto \mathbf r' = (x,(1-2p)y,(1-2p)z).
\end{equation*}

For the \(z\)-axis family

\begin{equation*}
\rho_r^{(\z)}=\frac{1}{2}(\I+r\,\sigma_z),
\end{equation*}

\noindent one gets

\begin{equation*}
\Phi^{BF}_p(\rho_r^{(\z)})
=
\frac{1}{2}\left(\I+(1-2p)r\,\sigma_z\right),
\end{equation*}

Since the input and output Bloch vectors are collinear, the corresponding
density matrices commute and share a common eigenbasis. In this case the
Procrustes optimization is trivial and yields

\begin{equation*}
S_\star=\I_3,
\end{equation*}

\noindent so that the misalignment angle $\Theta$ vanishes.

Therefore (cf. Appendix~\ref{app:collinear})

\begin{eqnarray*}
g_{\mathrm{BF}}(r,p)
=&
\frac{1}{2}
\left[
\sqrt{(1+r)(1+(1-2p)r)}\right.\\
&+
\left.\sqrt{(1-r)(1-(1-2p)r)}
\right],
\end{eqnarray*}

\begin{equation*}
D_N\!\left(\rho_r^{(\z)},\mathcal{B}_p(\rho_r^{(\z)})\right)
=
\sqrt{
\hbin\!\left(
\frac{1+g_{\mathrm{BF}}(r,p)}{2}
\right)
}.
\end{equation*}

\subsubsection{Phase-flip channel}

The phase-flip channel $\Phi^{PF}_p(\rho)$, with error probability \(p\), acts on the Bloch vector as(cf.~Sec.~\ref{sssec:phase_flip_channel})

\begin{equation*}
\mathbf r = (r_x,r_y,r_z)\mapsto \mathbf r' = ((1-2p)r_x,(1-2p)r_y,r_z).
\end{equation*}

For equatorial \(x\)-axis states, the input and output are collinear in Bloch space. Thus, the Procrustes rotation is $S_\star=\I_3$, and the misalignment angle $\Theta$ vanishes. Therefore (see Appendix~\ref{app:collinear})

\begin{eqnarray*}
g_{\mathrm{PF}}(r,p)
=&
\frac{1}{2}
\left[
\sqrt{(1+r)(1+(1-2p)r)}\right.\\
&+
\left.\sqrt{(1-r)(1-(1-2p)r)}
\right],
\end{eqnarray*}

\noindent and

\begin{equation*}
D_N\!\left(\rho_r^{(\x)},\mathcal{P}_p(\rho_r^{(\x)})\right)
=
\sqrt{
\hbin\!\left(
\frac{1+g_{\mathrm{PF}}(r,p)}{2}
\right)
}.
\end{equation*}

\subsubsection{General diagonal Pauli channel}

A diagonal Pauli channel has the Bloch action

\begin{equation*}
\mathbf r = (r_x,r_y,r_z)\mapsto \mathbf r' = (\lambda_x x,\lambda_y y,\lambda_z z),
\end{equation*}

\noindent subject to complete positivity constraints. For axis-adapted families

\begin{equation*}
\rho_r^{(i)}=\frac{1}{2}(\I+r\,\sigma_i),
\qquad i\in\{x,y,z\},
\end{equation*}

\noindent the output remains collinear with the input:

\begin{equation*}
\rho_r^{(i)} \mapsto \frac{1}{2}(\I+\lambda_i r\,\sigma_i),
\end{equation*}

\noindent the Procrustes rotation is $S_\star=\I_3$, and the misalignment angle $\Theta=0$. Therefore (cf. Appendix~\ref{app:collinear})

\begin{align*}
g_{\mathrm{Pauli},i}(r,\lambda_i)
=&
\frac{1}{2}
\left[
\sqrt{(1+r)(1+\lambda_i r)}\right.\\
&+
\left.\sqrt{(1-r)(1-\lambda_i r)}
\right],
\end{align*}

\begin{equation*}
D_N
=
\sqrt{
\hbin\!\left(
\frac{1+g_{\mathrm{Pauli},i}(r,\lambda_i)}{2}
\right)
}.
\end{equation*}

\subsection{Amplitude damping}
\label{ssec:amplitude}

The amplitude damping channel $\mathcal{A}_\gamma(\rho)$, with error probability \(p\), acts on the Bloch vector as (cf.~Sec.~\ref{sssec:amplitude_damping_channel})

\begin{equation*}
\mathbf r = (r_x,r_y,r_z)\mapsto \mathbf r'
\end{equation*}

\noindent where

\begin{equation*}
\mathbf r' = \left(
\sqrt{1-\gamma}\,r_x,\,
\sqrt{1-\gamma}\,r_y,\,
\gamma+(1-\gamma)r_z
\right).
\end{equation*}

For states diagonal in the energy basis,

\begin{equation*}
\rho_r^{(\z)}=\frac{1}{2}(\I+r\,\sigma_z),
\qquad -1\le r\le 1,
\end{equation*}

\noindent the output is

\begin{equation*}
\mathcal{A}_\gamma(\rho_r^{(\z)})
=
\frac{1}{2}\left[\I+r'\sigma_z\right],
\qquad
r'=\gamma+(1-\gamma)r.
\end{equation*}

Since the input and output states remain diagonal in the energy basis,
they are collinear in Bloch space. Hence, in the adapted frame,

\begin{equation*}
S_\star=\I_3,
\end{equation*}

\noindent and the misalignment angle $\Theta$ vanishes.
Therefore, the overlap is (cf. Appendix~\ref{app:collinear})

\begin{align*}
g_{\mathrm{AD}}(r,\gamma)
=&
\frac{1}{2}
\left[
\sqrt{(1+r)(1+r')}\right.\\
&+
\left.\sqrt{(1-r)(1-r')}
\right],\\
r'=&\gamma+(1-\gamma)r,
\end{align*}

and

\begin{equation*}
D_N\!\left(\rho_r^{(\z)},\mathcal{A}_\gamma(\rho_r^{(\z)})\right)
=
\sqrt{
\hbin\!\left(
\frac{1+g_{\mathrm{AD}}(r,\gamma)}{2}
\right)
}.
\end{equation*}

\section{Numerical implementation for generic channels}
\label{sec:numerical}

In addition to the analytically tractable families discussed above, it is instructive
to explore the behavior of the metric $D_N$ and the optimal purification alignment
$\Theta$ for generic qubit states and channels.  In this section we present a set of numerical
examples illustrating how the Fano--Procrustes framework reveals geometric
information that is not captured by scalar distinguishability measures alone.

\subsection{General algorithm for generic qubit channels}

Let \(\Phi\) be a qubit CPTP map. In Bloch form,

\begin{equation*}
\mathbf r \mapsto \mathbf r' = M\mathbf r + \mathbf c,
\end{equation*}

\noindent where \(M\) is a real \(3\times 3\) matrix and
\(\mathbf c\in\mathbb R^3\).

Given an input Bloch vector \(\mathbf r\), the numerical procedure
is as follows:

\begin{enumerate}[label=\arabic*.]

\item Compute the output Bloch vector

\begin{equation*}
\mathbf r' = M\mathbf r+\mathbf c.
\end{equation*}

\item Construct the canonical purification data
\((A_r,\boldsymbol{\gamma}_r)\) and
\((A_{r'},\widetilde{\boldsymbol{\delta}}_{r'})\), with

\begin{equation*}
A_r = O(\mathbf n_r)\diag(\sqrt{1-r^2},\sqrt{1-r^2},-1),
\qquad
\boldsymbol{\gamma}_r=r\,\hat{\mathbf z},
\end{equation*}

and analogously for \(\mathbf r'\).

\item Form the Procrustes matrix

\begin{equation*}
K = A_r^T A_{r'} + \boldsymbol{\gamma}_r\,\widetilde{\boldsymbol{\delta}}_{r'}^{\,T}.
\end{equation*}

\item Compute the singular value decomposition~\cite{Horn1986}

\begin{equation*}
K=U\Sigma V^T.
\end{equation*}

\item Obtain the optimal proper rotation

\begin{equation*}
S_\star = V\Lambda U^T,
\qquad
\Lambda=\diag(1,1,\det(VU^T)).
\end{equation*}

\item Form the aligned purification data

\begin{equation*}
B_\star = A_{r'} S_\star,
\qquad
\boldsymbol{\delta}_\star=S_\star^T\widetilde{\boldsymbol{\delta}}_{r'}.
\end{equation*}

\item Compute the optimized overlap from

\begin{equation*}
g_\star^2
=
\frac{1}{4}
\left[
1+\mathbf r\cdot\mathbf r'
+\boldsymbol{\gamma}_r\cdot\boldsymbol{\delta}_\star
+\Tr(A_r^T B_\star)
\right],
\end{equation*}

\noindent and then evaluate

\begin{equation*}
D_N(\rho,\Phi(\rho))
=
\sqrt{
\hbin\!\left(\frac{1+g_\star}{2}\right)
}.
\end{equation*}

\item Compute the purification misalignment angle through

\begin{equation*}
\Theta(\rho,\Phi(\rho))
=
\arccos\!\left(\frac{\Tr(S_\star)-1}{2}\right).
\end{equation*}

\end{enumerate}


\subsection{Bit-flip channel}
\label{ssec:bit_flip_numerical}

In the Bloch representation, the bit-flip channel $\Phi_p^{\mathrm{BF}}(\rho)$ acts as
(cf.~Sec.~\ref{sssec:bit_flip_channel})

\begin{equation*}
\mathbf r = (r_x,r_y,r_z)\mapsto \mathbf r' = (r_x,(1-2p)r_y,(1-2p)r_z).
\end{equation*}

The \(x\)-component is invariant, while the components perpendicular to
the \(x\)-axis are contracted by the factor \(\lambda\). For \(p=0\) the
channel reduces to the identity, whereas for \(p=1\) it becomes the
unitary bit-flip operation \(\rho\mapsto\sigma_x\rho\sigma_x\).

The bit-flip channel provides a useful example for the
Fano--Procrustes framework because its action is anisotropic in Bloch
space. For input states aligned with the \(x\)-axis, the input and
output Bloch vectors are collinear and the Procrustes rotation is
trivial. For generic input states, however, the output direction differs
from the input direction, and the optimal purification alignment is
generally nontrivial.

Figure~\ref{fig:bf_DN_radii} shows the metric $D_N$ as a function of
\(p\) for a fixed pair of Bloch angles $(\phi,\theta)$ and several values of the
Bloch-vector length $r=\|\mathbf r\|$. The dependence on the radius reflects the fact
that states closer to the surface of the Bloch ball are more sensitive
to bit-flip noise. The metric vanishes at \(p=0\), as expected, and
increases as the output state departs from the input state.

\begin{figure}[t]
\centering
\includegraphics[width=0.9\linewidth]{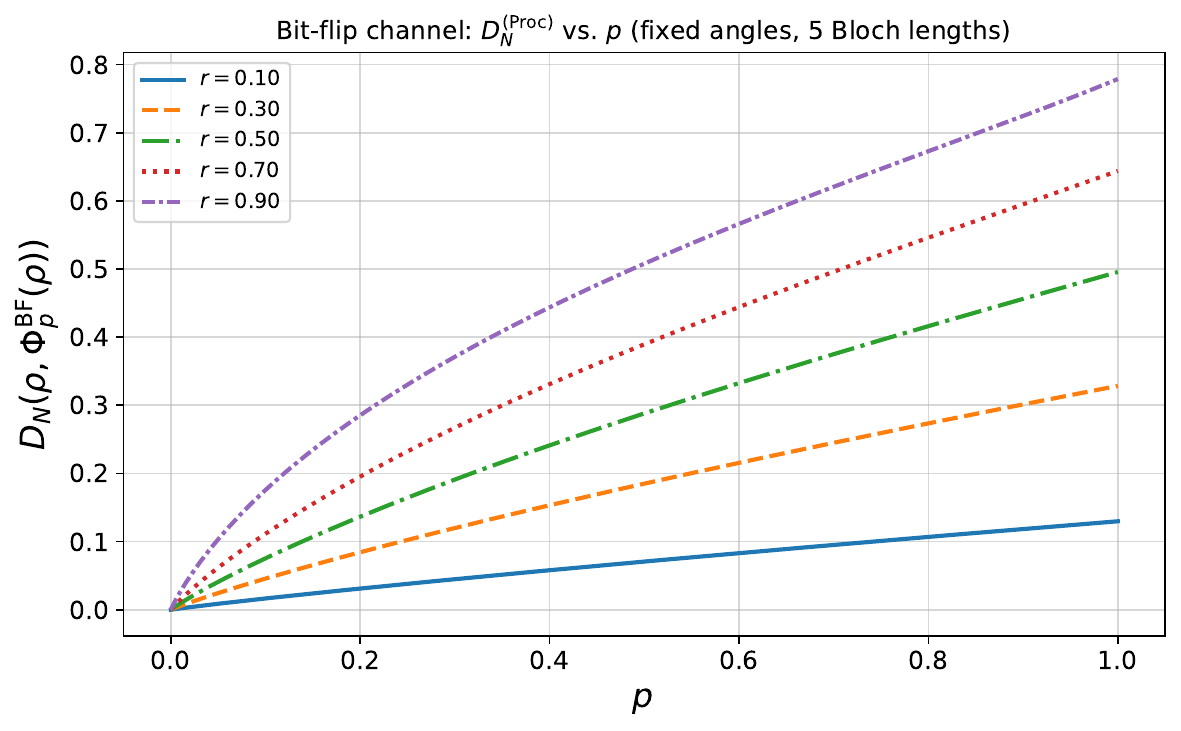}
\caption{The metric
\(D_N(\rho,\Phi_p^{\mathrm{BF}}(\rho))\) as a function of the bit-flip
probability \(p\), for fixed Bloch angles and different values of the
Bloch-vector length.}
\label{fig:bf_DN_radii}
\end{figure}

The corresponding purification misalignment angle is shown in
Fig.~\ref{fig:bf_Theta_radii}. In contrast with the scalar metric,
\(\Theta\) is sensitive to the change in orientation of the Bloch vector
induced by the anisotropic contraction of the \(y\) and \(z\)
components. For generic input states, this produces a nonzero optimal
rotation of the purification frame.

\begin{figure}[t]
\centering
\includegraphics[width=0.9\linewidth]{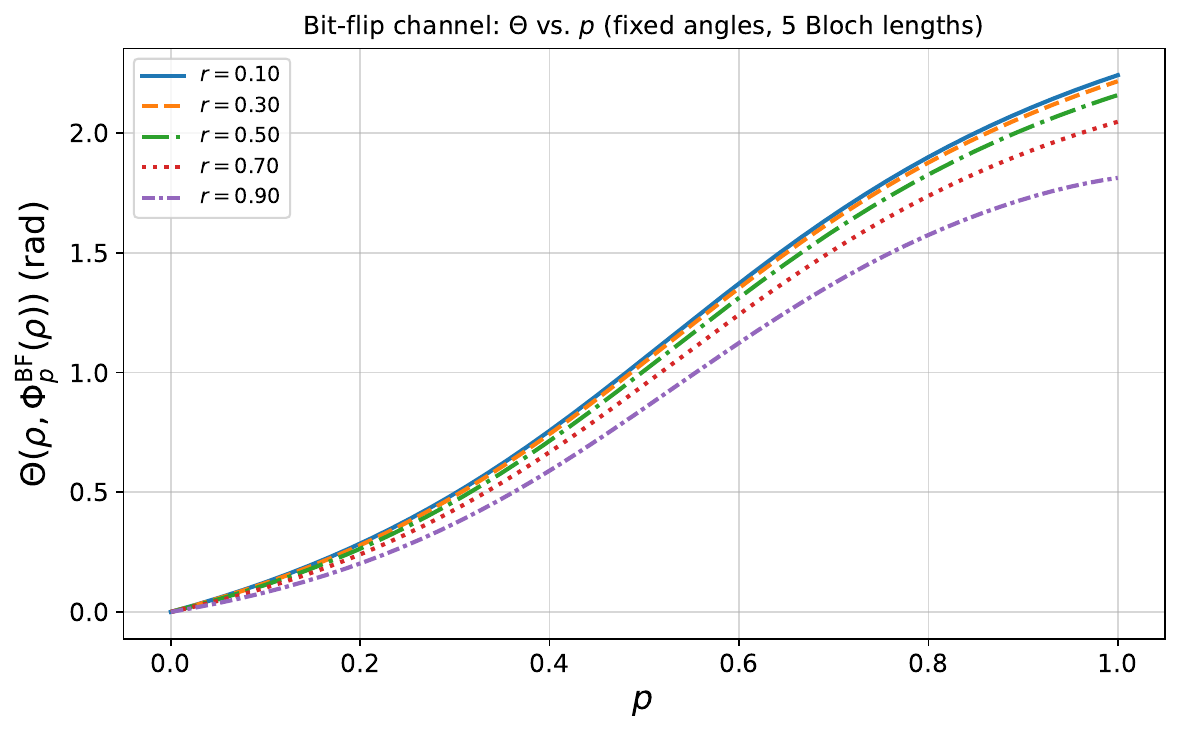}
\caption{Purification misalignment angle
\(\Theta(\rho,\Phi_p^{\mathrm{BF}}(\rho))\) as a function of the
bit-flip probability \(p\), for the same configurations as in
Fig.~\ref{fig:bf_DN_radii}. The angle captures the geometric
realignment of the optimal purification frames induced by the
anisotropic channel action.}
\label{fig:bf_Theta_radii}
\end{figure}

We next consider the dependence on the orientation of the input state.
Figure~\ref{fig:bf_DN_angles} shows
\(D_N(\rho,\Phi_p^{\mathrm{BF}}(\rho))\) for a fixed Bloch-vector length
and several choices of pairs of Bloch angles $(\phi,\theta)$. Since the bit-flip channel leaves
the \(x\)-axis invariant while contracting the transverse directions,
the scalar distinguishability depends strongly on how much of the input
Bloch vector lies in the plane perpendicular to the \(x\)-axis.

\begin{figure}[t]
\centering
\includegraphics[width=0.9\linewidth]{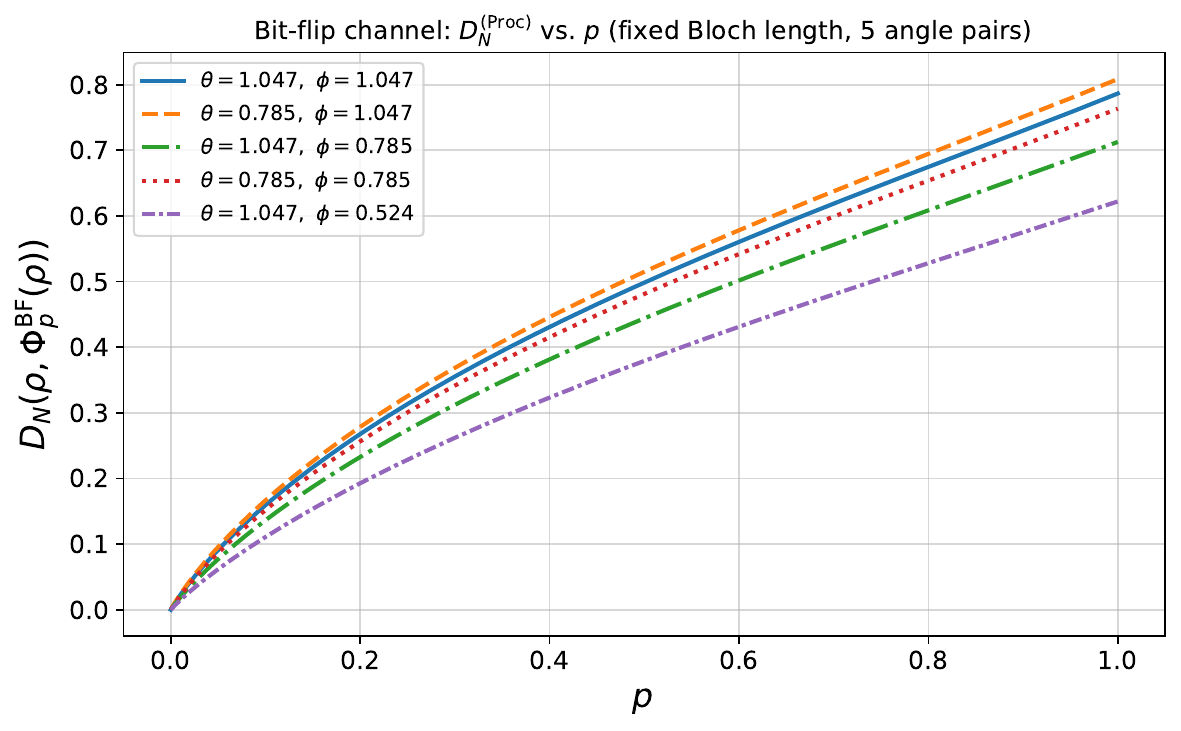}
\caption{The metric
\(D_N(\rho,\Phi_p^{\mathrm{BF}}(\rho))\) as a function of \(p\), for a
fixed Bloch-vector length and different orientations of the input
state.}
\label{fig:bf_DN_angles}
\end{figure}

The corresponding behavior of the misalignment angle is displayed in
Fig.~\ref{fig:bf_Theta_angles}. The dependence on the input direction
is especially relevant for \(\Theta\), because the optimal Procrustes
rotation detects the orientation change of the purification frame. For
states aligned with the symmetry axis of the channel, the action is
purely radial or trivial and the misalignment angle vanishes. Away from
that axis, the anisotropic contraction generally produces a nonzero
\(\Theta\).

\begin{figure}[t]
\centering
\includegraphics[width=0.9\linewidth]{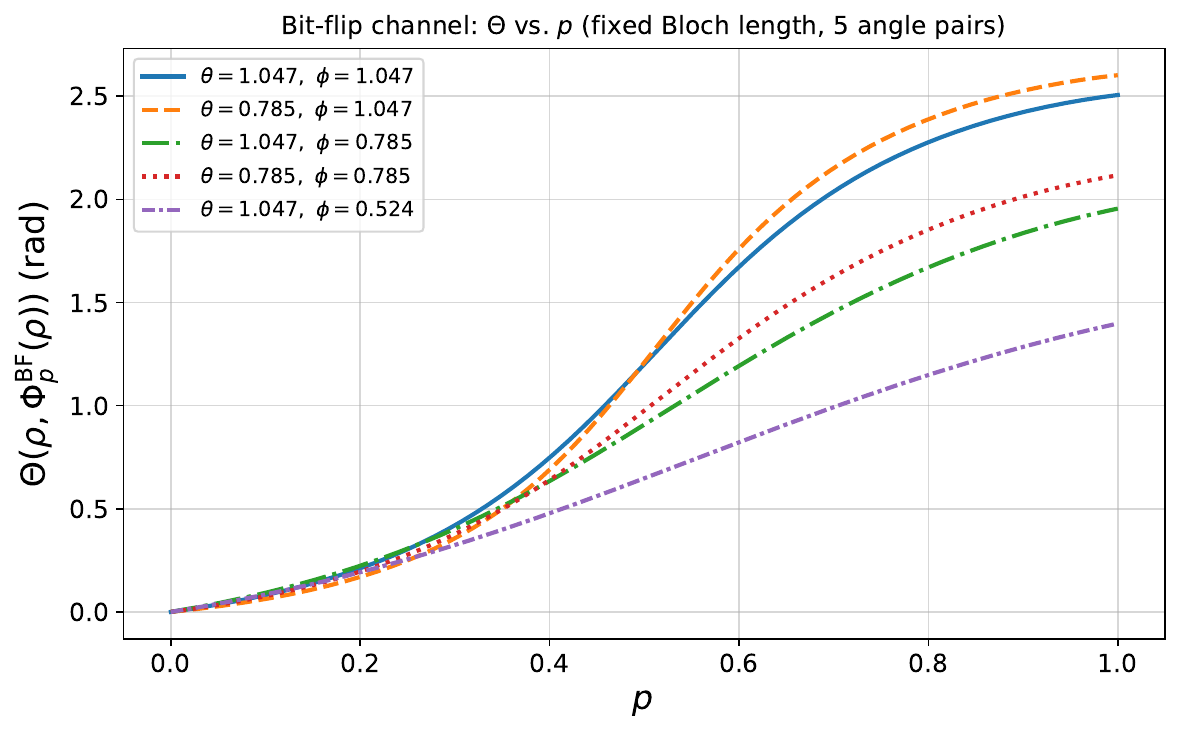}
\caption{Purification misalignment angle
\(\Theta(\rho,\Phi_p^{\mathrm{BF}}(\rho))\) as a function of \(p\), for
the same configurations as in Fig.~\ref{fig:bf_DN_angles}. The
orientation dependence reflects the anisotropic nature of the bit-flip
channel.}
\label{fig:bf_Theta_angles}
\end{figure}

The results show that \(D_N\) and \(\Theta\) provide complementary
information about the action of the bit-flip channel. The metric
\(D_N\) quantifies the magnitude of the change in the state, while
\(\Theta\) characterizes the geometric reorientation required to align
the optimal purification frames.


\subsection{Phase-flip channel}
\label{ssec:phase_flip_numerical}

The phase-flip channel $\Phi_p^{\mathrm{PF}}(\rho)$ acts on the Bloch vector as (cf.~Sec.~\ref{sssec:phase_flip_channel})

\begin{equation*}
\mathbf r=(r_x,r_y,r_z)
\longmapsto
\mathbf r'
=
\bigl(\lambda r_x,\lambda r_y,r_z\bigr),
\qquad
\lambda=1-2p .
\end{equation*}

The \(z\)-component is invariant, while the coherences, represented by
the transverse components \(r_x\) and \(r_y\), are contracted by the
factor \(\lambda\). For \(p=0\) the channel is the identity, while for
\(p=1\) it becomes the unitary phase flip
\(\rho\mapsto\sigma_z\rho\sigma_z\). At \(p=1/2\), the transverse
components vanish and the channel acts as complete dephasing in the
\(\sigma_z\) eigenbasis.

The phase-flip channel is a useful example for the Fano--Procrustes
framework because it is anisotropic but unital. It singles out the
\(z\)-axis as a symmetry axis. For input states aligned with the
\(z\)-axis, the input and output Bloch vectors coincide and the
Procrustes rotation is trivial. For generic input states, however, the
transverse contraction changes the direction of the Bloch vector, and
the optimal purification alignment can become nontrivial.

Figure~\ref{fig:pf_DN_radii} shows the metric $D_N$ as a function of
\(p\), for a fixed pair of Bloch angles and several values of the
Bloch-vector length. The metric vanishes at \(p=0\), where the channel
is the identity, and increases as the phase-flip probability grows. The
dependence on the Bloch-vector length reflects the fact that states
with larger radius, and hence higher purity, are more sensitive to the
loss of transverse coherence.

\begin{figure}[t]
\centering
\includegraphics[width=0.9\linewidth]{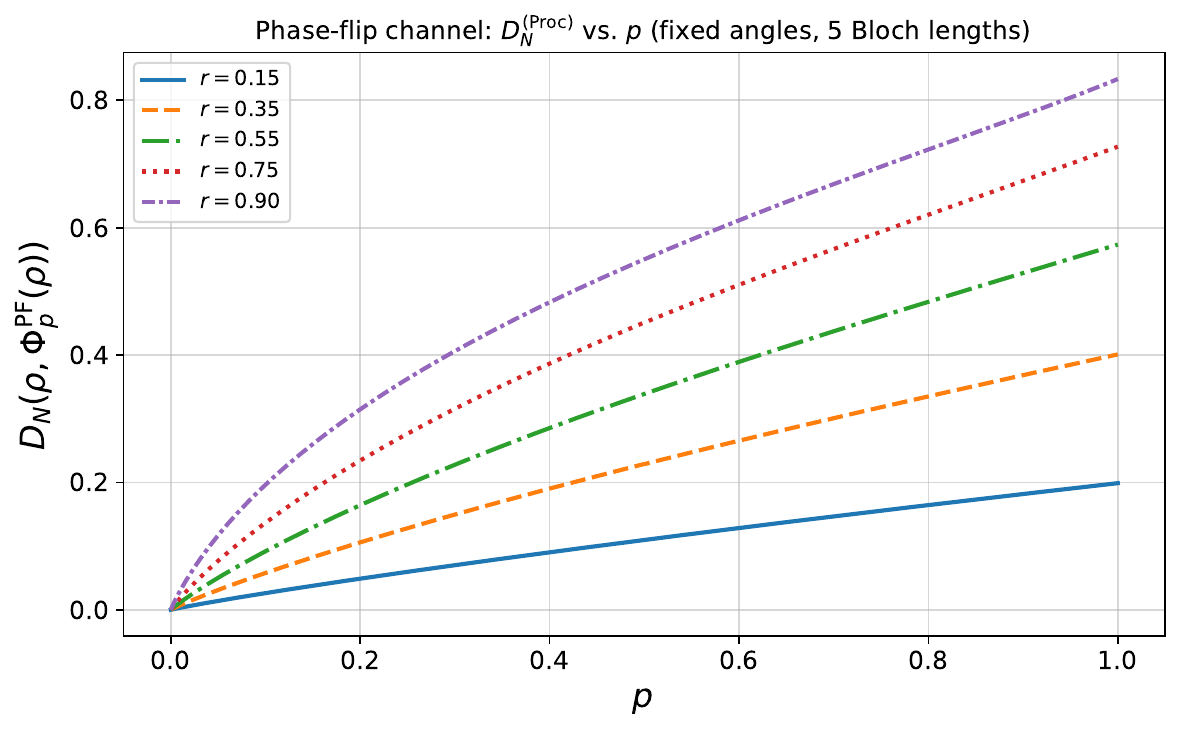}
\caption{The metric 
\(D_N(\rho,\Phi_p^{\mathrm{PF}}(\rho))\) as a function of the
phase-flip probability \(p\), for fixed Bloch angles and different
values of the Bloch-vector length.}
\label{fig:pf_DN_radii}
\end{figure}

The corresponding purification misalignment angle is shown in
Fig.~\ref{fig:pf_Theta_radii}. Since the phase-flip channel contracts
the components perpendicular to the \(z\)-axis while leaving the
\(z\)-component unchanged, the output Bloch direction generally differs
from the input direction. This produces a nontrivial optimal
Procrustes rotation and hence a nonzero value of \(\Theta\) for generic
input states.

\begin{figure}[t]
\centering
\includegraphics[width=0.9\linewidth]{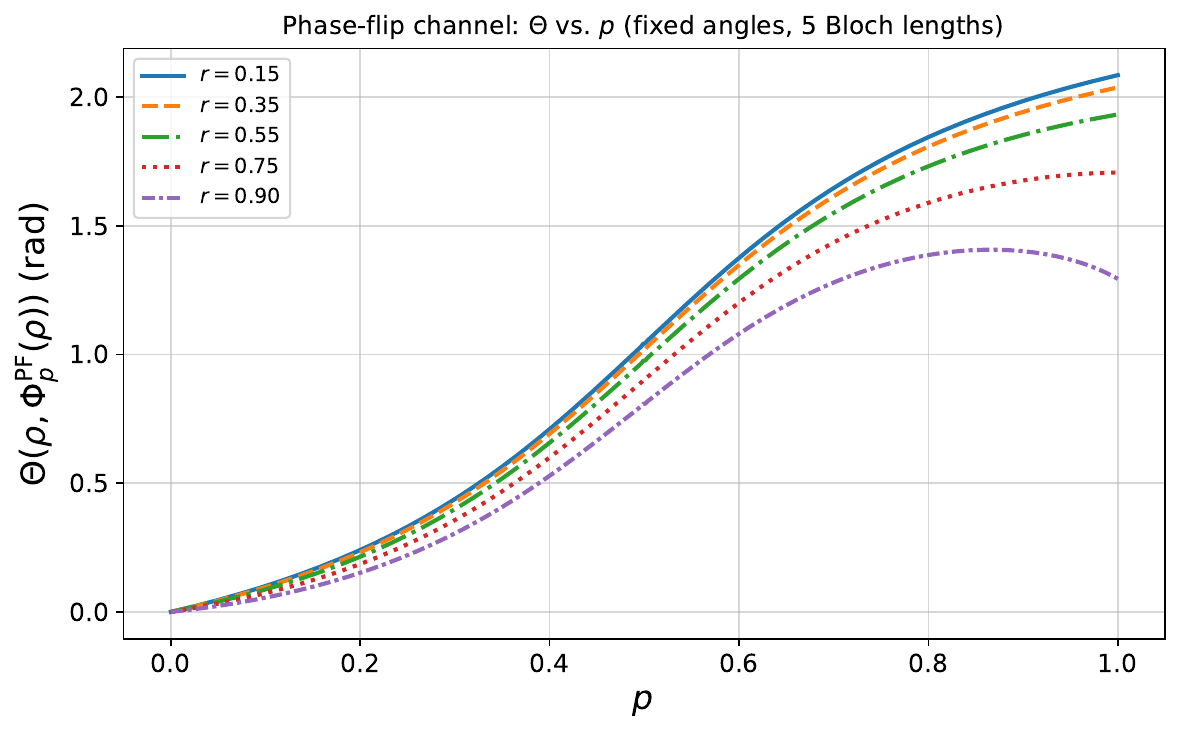}
\caption{Purification misalignment angle
\(\Theta(\rho,\Phi_p^{\mathrm{PF}}(\rho))\) as a function of the
phase-flip probability \(p\), for the same configurations as in
Fig.~\ref{fig:pf_DN_radii}. The angle detects the geometric
realignment of the optimal purification frames induced by the
anisotropic contraction of the transverse Bloch components.}
\label{fig:pf_Theta_radii}
\end{figure}

We next examine the dependence on the orientation of the input state.
Figure~\ref{fig:pf_DN_angles} shows
\(D_N(\rho,\Phi_p^{\mathrm{PF}}(\rho))\) for a fixed Bloch-vector length
and several choices of Bloch angles. Because the channel preserves the
\(z\)-axis and suppresses the transverse components, the scalar
distinguishability depends on the relative weight of the input Bloch
vector in the transverse plane. States close to the \(z\)-axis are less
affected, while states with larger transverse components experience a
larger change.

\begin{figure}[t]
\centering
\includegraphics[width=0.9\linewidth]{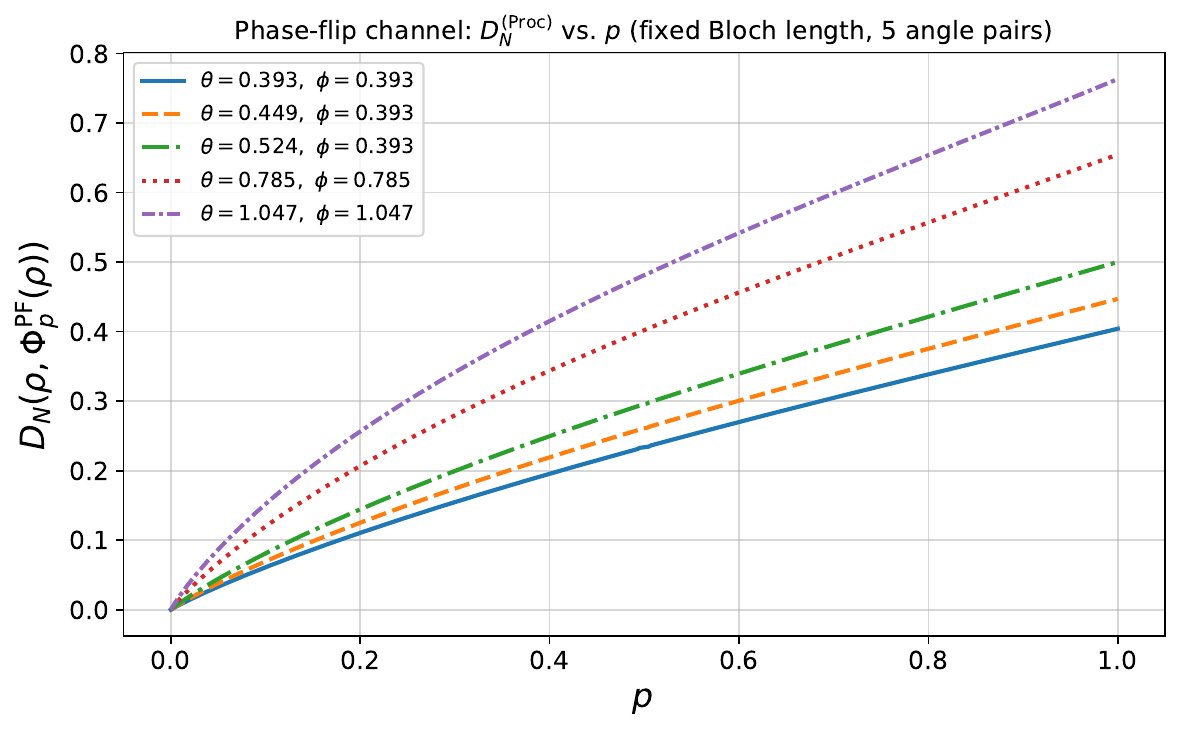}
\caption{The metric
\(D_N(\rho,\Phi_p^{\mathrm{PF}}(\rho))\) as a function of \(p\), for a
fixed Bloch-vector length and different orientations of the input
state.}
\label{fig:pf_DN_angles}
\end{figure}

The corresponding behavior of the misalignment angle is displayed in
Fig.~\ref{fig:pf_Theta_angles}. The angle vanishes for states aligned
with the symmetry axis of the channel, since those states remain fixed
by the phase-flip dynamics. Away from this axis, the anisotropic
contraction changes the Bloch direction and the Procrustes optimization
yields a nontrivial rotation. Thus, \(\Theta\) provides a geometric
indicator of the nonradial component of the phase-flip channel action.

\begin{figure}[t]
\centering
\includegraphics[width=0.9\linewidth]{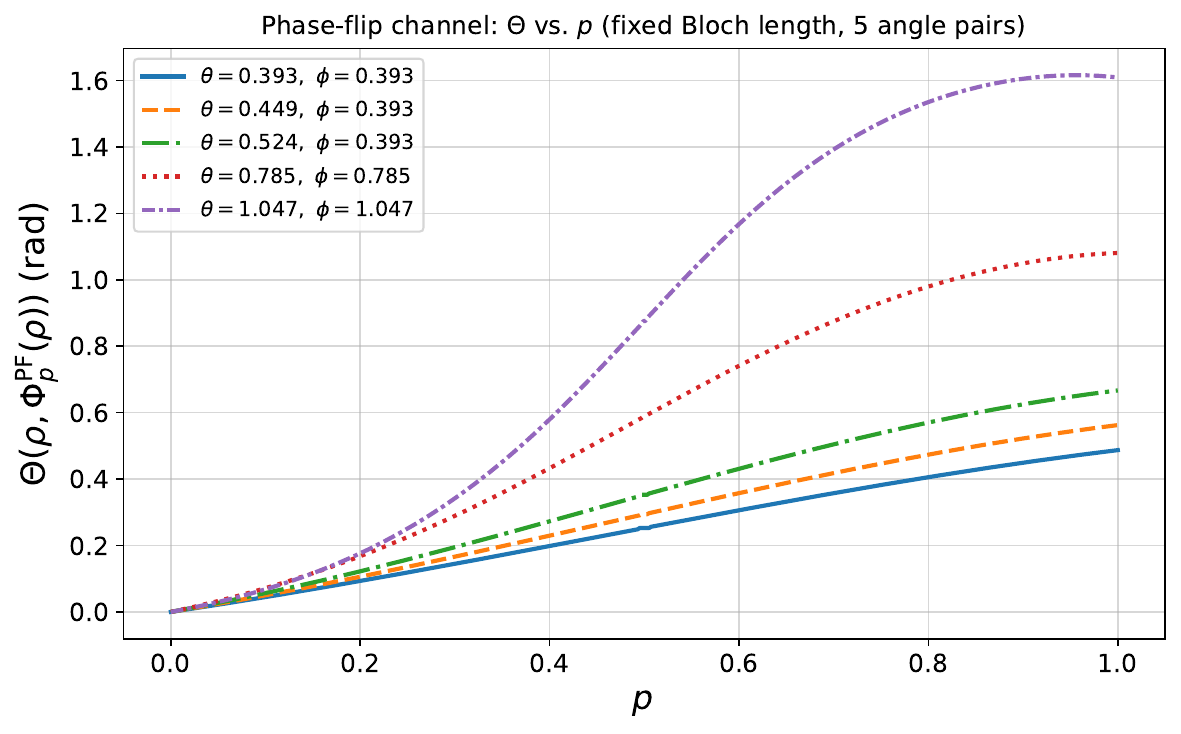}
\caption{Purification misalignment angle
\(\Theta(\rho,\Phi_p^{\mathrm{PF}}(\rho))\) as a function of \(p\), for
the same configurations as in Fig.~\ref{fig:pf_DN_angles}. The
orientation dependence reflects the axial symmetry of the phase-flip
channel and its anisotropic action on the Bloch ball.}
\label{fig:pf_Theta_angles}
\end{figure}

The results show that the metric $D_N$ and the Procrustes angle $\Theta$
capture complementary aspects of the phase-flip dynamics. The metric
\(D_N\) quantifies the magnitude of the state change produced by the
loss of coherence, while \(\Theta\) captures the geometric
reorientation of the optimal purification frames. In particular, for
states diagonal in the \(\sigma_z\) basis one has
\(\mathbf r'=\mathbf r\), so both the metric and the misalignment angle
vanish. For generic states, however, the transverse contraction induces
both scalar distinguishability and nontrivial purification-frame
misalignment. This makes the phase-flip channel a simple benchmark for
distinguishing radial changes from anisotropic geometric effects in the
Fano--Procrustes framework.


\subsection{Amplitude damping channel}
\label{ssec:amplitude_damping_numerical}

In the Bloch representation, the amplitude damping channel $\mathcal{A}_\gamma(\rho)$ acts as the affine transformation (cf.~Sec.~\ref{sssec:amplitude_damping_channel})

\begin{equation*}
\mathbf r=(r_x,r_y,r_z)
\longmapsto \mathbf r'
\end{equation*}

\noindent where

\begin{equation*}
\mathbf r'
=
\left(
\sqrt{1-\gamma}\,r_x,\,
\sqrt{1-\gamma}\,r_y,\,
\gamma+(1-\gamma)r_z
\right).
\end{equation*}

Unlike Pauli channels, amplitude damping is nonunital, i.e., it does not leave
the maximally mixed state fixed. The transverse components \(r_x\) and
\(r_y\) are contracted by the factor \(\sqrt{1-\gamma}\), while the
longitudinal component is both contracted and translated toward the
positive \(z\)-axis. In particular, the ground state \(\ket{0}\) is a
fixed point of the channel.

The amplitude damping channel is especially useful for testing the
Fano--Procrustes framework because it combines contraction with an
affine drift. For states diagonal in the energy basis, the input and
output Bloch vectors remain collinear and the Procrustes rotation is
trivial. For generic input states, however, the affine displacement
changes the direction of the Bloch vector, and the optimal purification
alignment can become nontrivial.

Figure~\ref{fig:ad_DN_radii} shows the metric $D_N$ as a function of
\(\gamma\), for a fixed pair of Bloch angles and several values of the
Bloch-vector length. The metric vanishes at \(\gamma=0\), where the
channel reduces to the identity map. As \(\gamma\) increases, the output
state moves toward the ground state and the value of \(D_N\) grows,
reflecting the increasing distinguishability between the input and
output states. The dependence on the Bloch radius indicates that states
with different purities respond differently to the dissipative
relaxation.

\begin{figure}[t]
\centering
\includegraphics[width=0.9\linewidth]{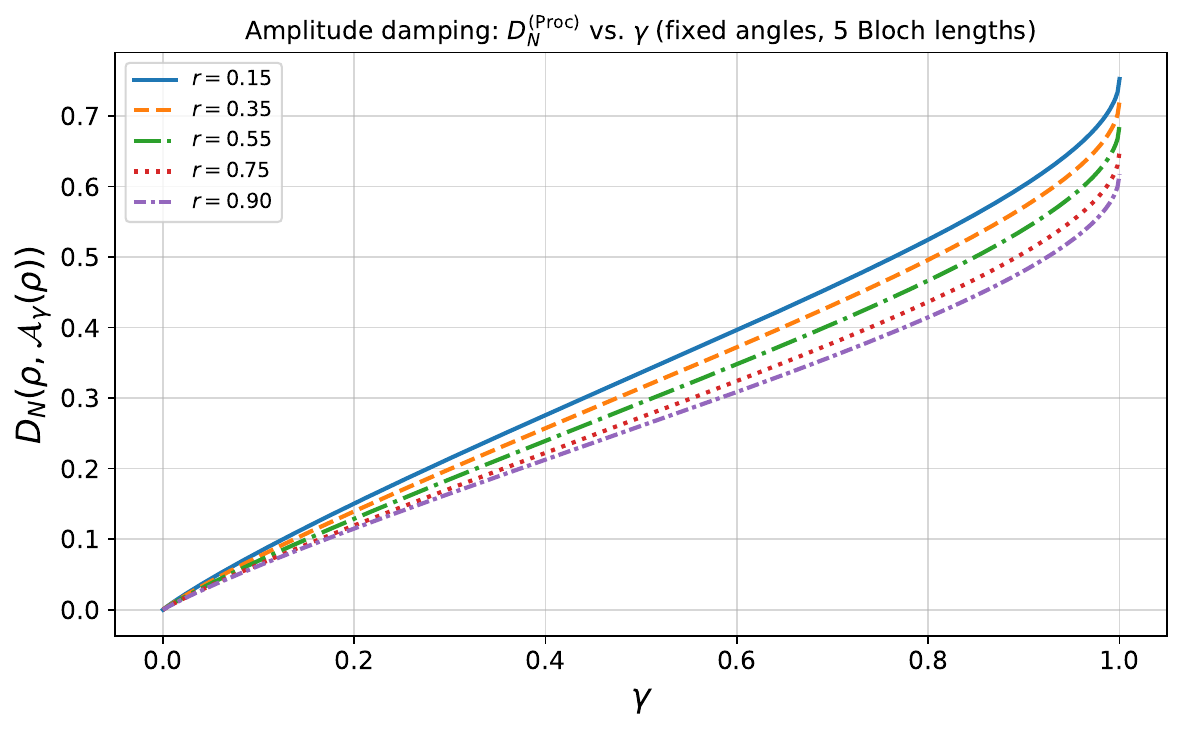}
\caption{The metric
\(D_N(\rho,\mathcal{A}_\gamma(\rho))\) as a function of the damping
parameter \(\gamma\), for fixed Bloch angles and different values of the
Bloch-vector length.}
\label{fig:ad_DN_radii}
\end{figure}

The corresponding purification misalignment angle is shown in
Fig.~\ref{fig:ad_Theta_radii}. Since amplitude damping is not purely
radial for generic input states, the output Bloch direction generally
differs from the input direction. This produces a nontrivial optimal
Procrustes rotation and hence a nonzero value of \(\Theta\). The angle
therefore captures the geometric component of the dissipative dynamics,
namely the purification-frame reorientation induced by the affine drift
and transverse contraction.

\begin{figure}[t]
\centering
\includegraphics[width=0.9\linewidth]{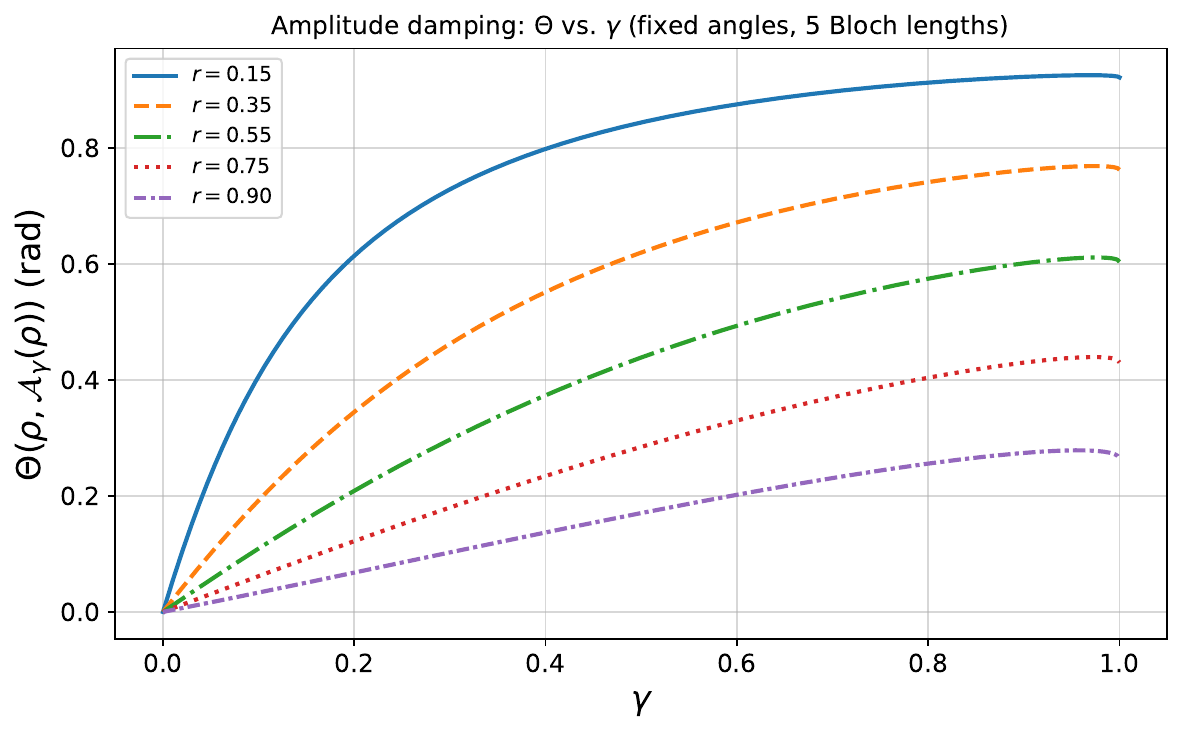}
\caption{Purification misalignment angle
\(\Theta(\rho,\mathcal{A}_\gamma(\rho))\) as a function of the damping
parameter \(\gamma\), for the same configurations as in
Fig.~\ref{fig:ad_DN_radii}. The angle $\Theta$ detects the geometric reorientation of the optimal purification frames induced by the
nonunital channel action.}
\label{fig:ad_Theta_radii}
\end{figure}

We next examine the dependence on the orientation of the input Bloch
vector. Figure~\ref{fig:ad_DN_angles} shows
\(D_N(\rho,\mathcal{A}_\gamma(\rho))\) for a fixed Bloch-vector length
and several choices of Bloch angles. Because amplitude damping singles
out the \(z\)-axis, the scalar distinguishability depends strongly on
the position of the input state relative to the energy basis. States
with different values of \(r_z\) experience different amounts of
longitudinal displacement and relaxation.

\begin{figure}[t]
\centering
\includegraphics[width=0.9\linewidth]{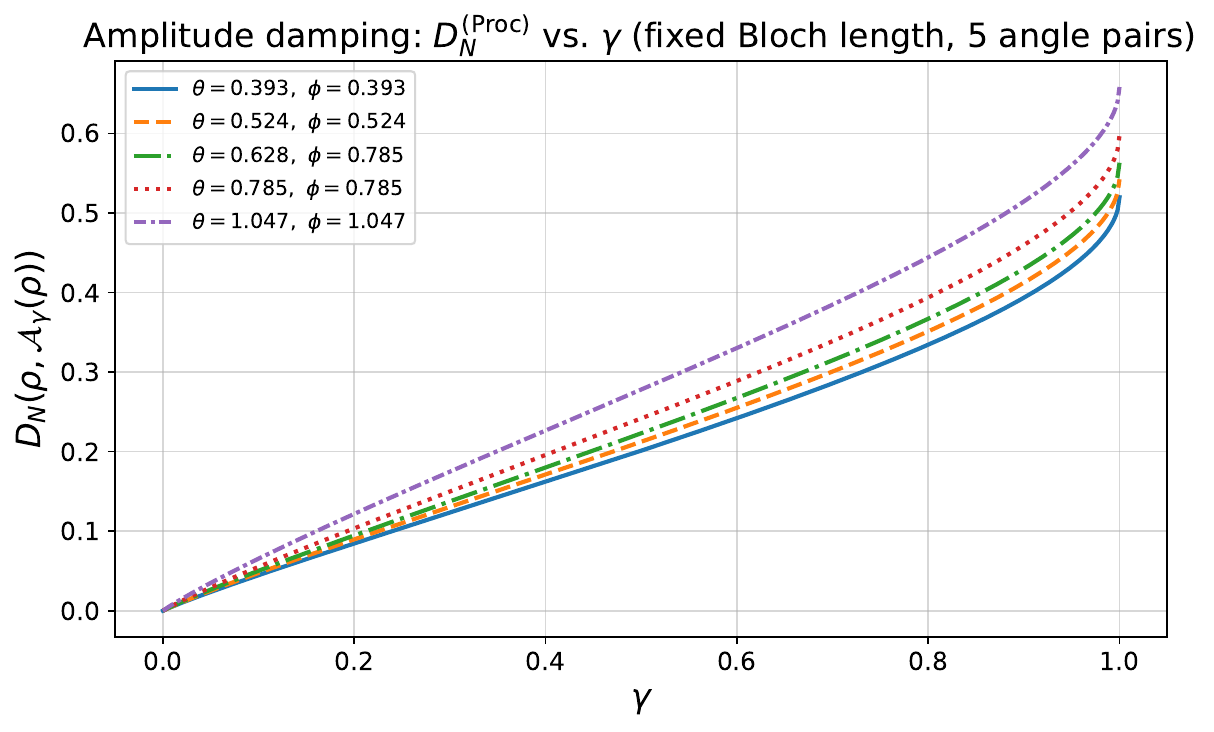}
\caption{The metric
\(D_N(\rho,\mathcal{A}_\gamma(\rho))\) as a function of \(\gamma\), for a
fixed Bloch-vector length and different orientations of the input
state.}
\label{fig:ad_DN_angles}
\end{figure}

The corresponding behavior of \(\Theta\) is displayed in
Fig.~\ref{fig:ad_Theta_angles}. The directional dependence is especially
relevant for this channel, i.e., when the input state is aligned with the
\(z\)-axis, the evolution remains collinear and the misalignment angle
vanishes. Away from the damping axis, however, the channel changes the
direction of the Bloch vector, and the Procrustes optimization yields a
nontrivial rotation. Thus \(\Theta\) provides a direct geometric
indicator of the nonradial part of the amplitude damping dynamics.

\begin{figure}[t]
\centering
\includegraphics[width=0.9\linewidth]{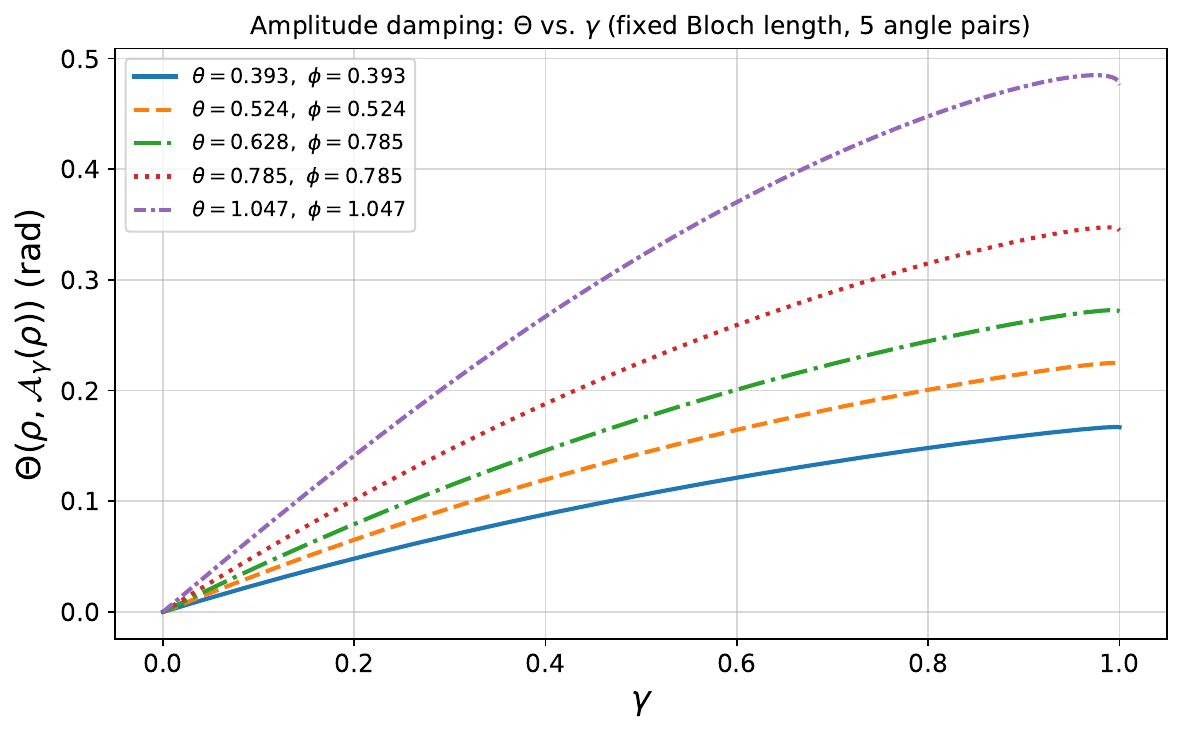}
\caption{Purification misalignment angle
\(\Theta(\rho,\mathcal{A}_\gamma(\rho))\) as a function of \(\gamma\),
for the same configurations as in Fig.~\ref{fig:ad_DN_angles}. The
orientation dependence reflects the fact that amplitude damping selects
a preferred axis in the Bloch ball.}
\label{fig:ad_Theta_angles}
\end{figure}

The results show that the metric $D_N$ and the Procrustes angle $\Theta$
highlight different aspects of the amplitude damping process. The metric
\(D_N\) quantifies the magnitude of the dissipative change, whereas
\(\Theta\) captures the geometric reorientation of the optimal
purification frames. For states diagonal in the energy basis, the
channel acts along a fixed Bloch direction and \(\Theta=0\). For generic
states, the affine drift toward the ground state produces a nonzero
misalignment angle. In this sense, the Fano--Procrustes framework
separates the scalar effect of relaxation from the geometric effect of
direction-changing dynamics.

\subsection{Imperfect quantum NOT gate}
\label{ssec:imperfect_not_numerical}

In the Bloch representation, the imperfect quantum NOT gate $\Phi^{\mathrm{NOT}}_{p}(\rho)$ 
acts on the Bloch vector as (cf. Sec.~\ref{ssec:imperfect_not})

\[
\mathbf r = (r_x,r_y,r_z)  \longmapsto \mathbf r'
\]

where 

\begin{equation*}
\mathbf{r}' =
\begin{pmatrix}
r_x \\
-\left[ (1-p) + p\cos(\delta\alpha) \right] r_y + p \sin(\delta\alpha)\, r_z \\
-\left[ (1-p) + p\cos(\delta\alpha) \right] r_z - p \sin(\delta\alpha)\, r_y
\end{pmatrix}.
\end{equation*}

This channel (cf.~Eq.~\eqref{eq:imperfect_not_channel} in Sec.~\ref{ssec:imperfect_not}) represents a random unitary
process arising from control imperfections in the implementation of a
$\pi$ pulse. Such errors are common, for instance, in Nuclear Magnetic
Resonance (NMR) experiments, where radiofrequency pulses are used to
manipulate spin-$1/2$ systems. If the applied pulse deviates slightly
from the ideal $\pi$ rotation, the resulting state will not be perfectly
inverted and may acquire additional coherences.

\medskip

\noindent
This model provides a simple yet realistic description of coherent
control errors and will be useful for analyzing the behavior of the
metric $D_N(\rho,\sigma)$ and the misalignment angle $\Theta(\rho,\sigma)$ under imperfect
unitary operations.

Figure~\ref{fig:not_error_DN_radii} shows the metric
\(D_N(\sigma_x\rho\sigma_x,\Phi^{\mathrm{NOT}}_{p}(\rho))\) as a function of the
noise parameter \(p\), for a fixed pair of Bloch angles and several
values of the Bloch-vector length. The corresponding misalignment angle \(\Theta(\sigma_x\rho\sigma_x,\Phi^{\mathrm{NOT}}_{p}(\rho))\) is shown in 
Figure~\ref{fig:not_error_Theta_radii}. In this figures, the actual output of the imperfect NOT gate is
compared with the ideal NOT output. In this case the quantities vanish
when \(p=0\), since the channel coincides with the ideal NOT operation.
As \(p\) increases, the distance
\(D_N(\sigma_x\rho\sigma_x,\Phi^{\mathrm{NOT}}_{p}(\rho))\)
measures the scalar loss of agreement with the target output, while
\(\Theta(\sigma_x\rho\sigma_x,\Phi^{\mathrm{NOT}}_{p}(\rho))\)
measures the geometric rotation required to align the corresponding
optimal purifications.

\begin{figure}[t]
\centering
\includegraphics[width=0.9\linewidth]{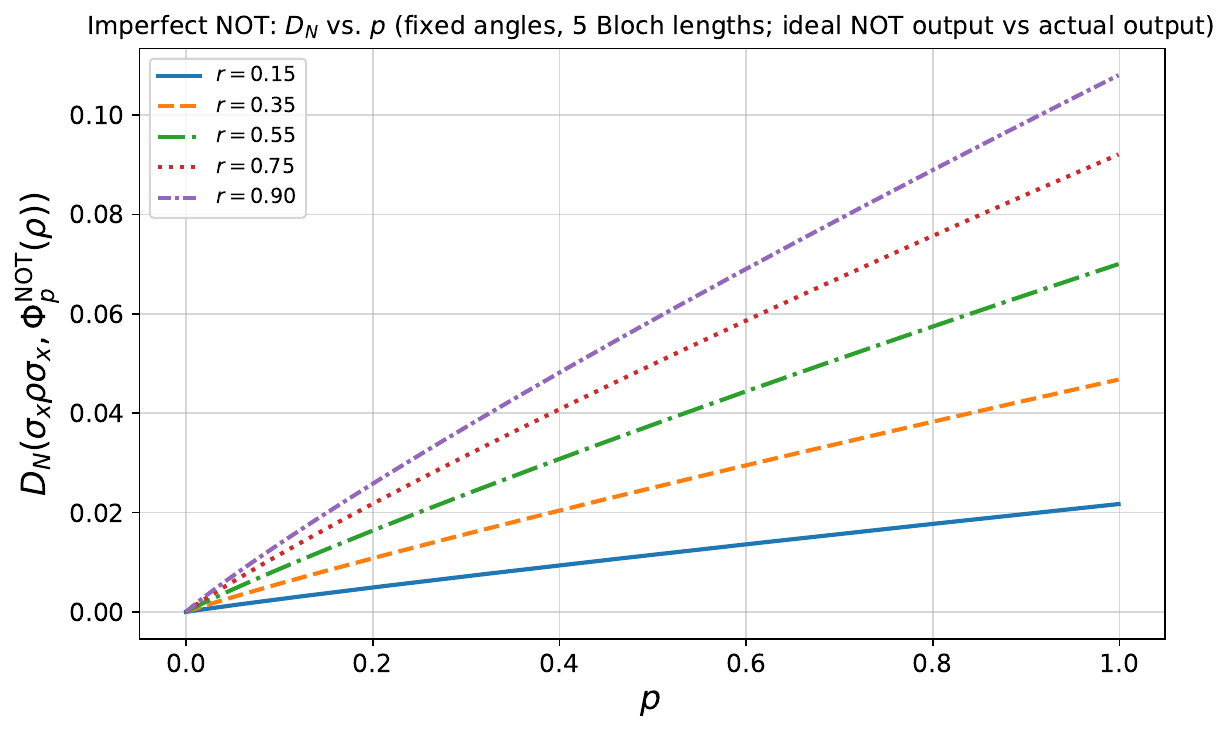}
\caption{The metric
\(D_N(\sigma_x\rho\sigma_x,\Phi^{\mathrm{NOT}}_{p}(\rho))\) as a
function of \(p\), for fixed Bloch angles and different Bloch-vector
lengths. This quantity isolates the deviation of the imperfect NOT
channel from the ideal NOT gate.}
\label{fig:not_error_DN_radii}
\end{figure}

\begin{figure}[t]
\centering
\includegraphics[width=0.9\linewidth]{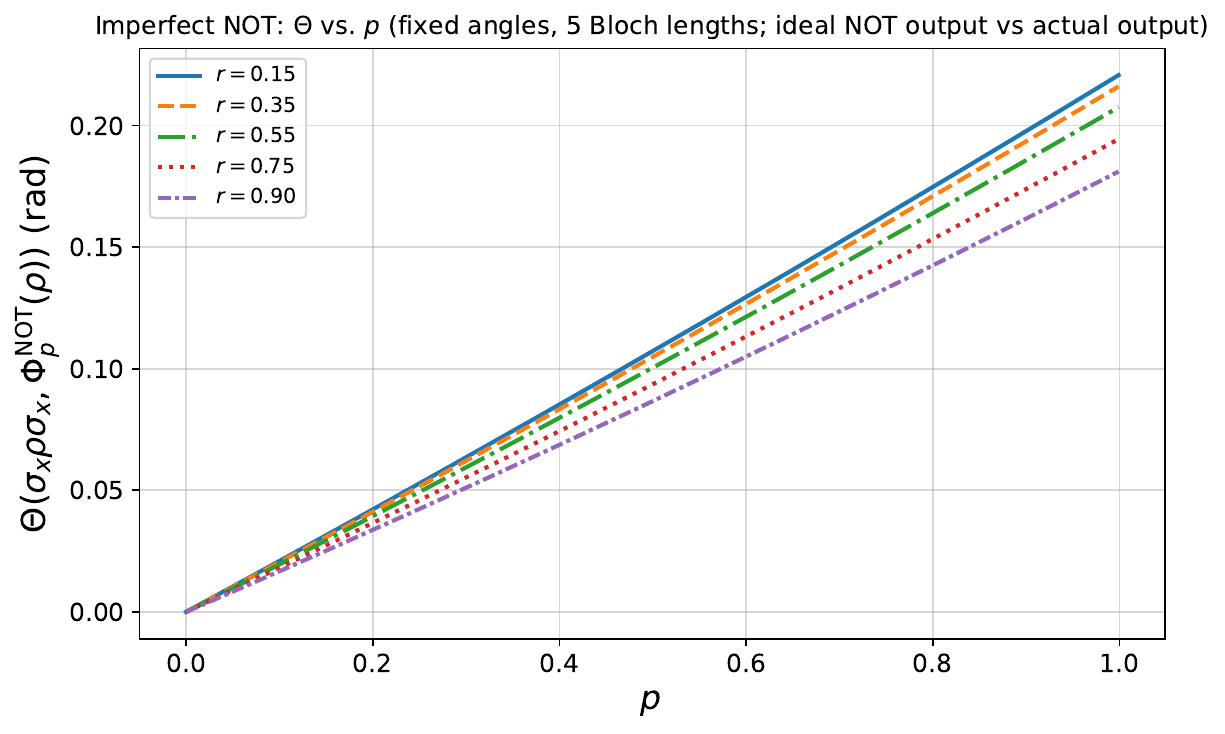}
\caption{Purification misalignment angle
\(\Theta(\sigma_x\rho\sigma_x,\Phi^{\mathrm{NOT}}_{p}(\rho))\) as a
function of \(p\), for the same configurations as in
Fig.~\ref{fig:not_error_DN_radii}. This angle provides a geometric
diagnostic of the coherent pulse-angle error relative to the target NOT
operation.}
\label{fig:not_error_Theta_radii}
\end{figure}

We also analyze the dependence on the orientation of the input state.
For a fixed Bloch-vector length and several choices of the Bloch angles,
Fig.~\ref{fig:not_error_DN_angles} shows the metric \(D_N(\sigma_x\rho\sigma_x,\Phi^{\mathrm{NOT}}_{p}(\rho))\),
 while Fig.~\ref{fig:not_error_Theta_angles} shows the corresponding
misalignment angle \(\Theta(\sigma_x\rho\sigma_x,\Phi^{\mathrm{NOT}}_{p}(\rho))\). 
These plots separate the effect of the intended
NOT operation from the additional error due to the imperfect pulse. This
comparison is especially useful for characterizing gate performance,
because it directly measures the deviation between the implemented
operation and the desired target transformation.

\begin{figure}[t]
\centering
\includegraphics[width=0.9\linewidth]{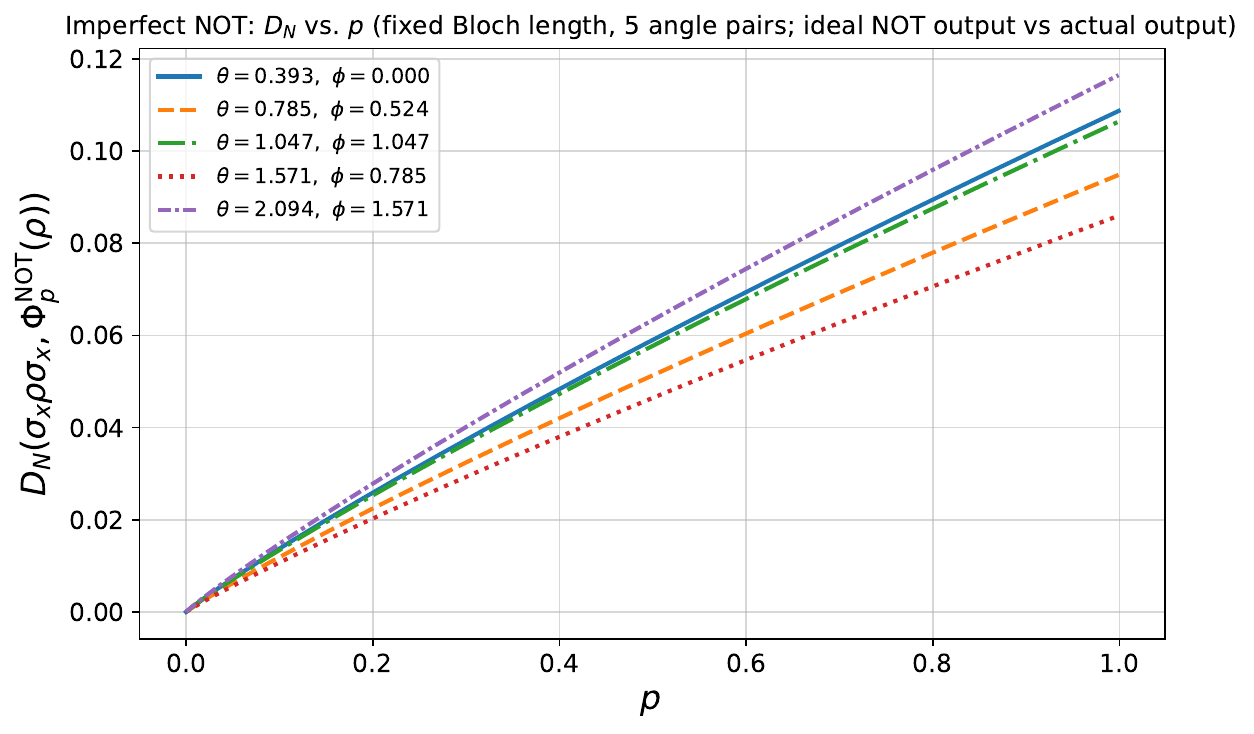}
\caption{The metric
\(D_N(\sigma_x\rho\sigma_x,\Phi^{\mathrm{NOT}}_{p}(\rho))\) as a
function of \(p\), for fixed Bloch-vector length and different
input-state orientations.}
\label{fig:not_error_DN_angles}
\end{figure}

\begin{figure}[t]
\centering
\includegraphics[width=0.9\linewidth]{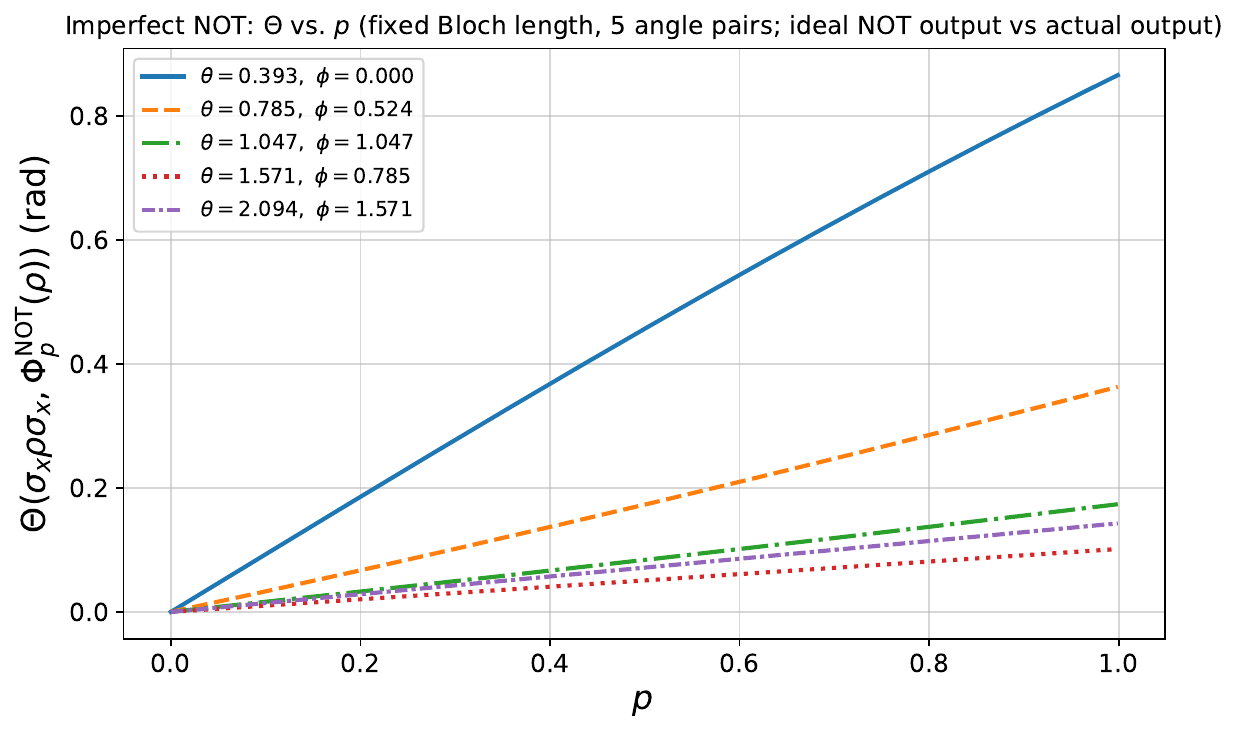}
\caption{Purification misalignment angle
\(\Theta(\sigma_x\rho\sigma_x,\Phi^{\mathrm{NOT}}_{p}(\rho))\) as a
function of \(p\), for the same configurations as in
Fig.~\ref{fig:not_error_DN_angles}. This quantity captures the
orientation-dependent geometric error induced by the miscalibrated NOT
pulse.}
\label{fig:not_error_Theta_angles}
\end{figure}

The results show that the metric $D_N$ and the misalignment angle $\Theta$
provide complementary information. The metric \(D_N\) quantifies the
magnitude of the state or gate-output deviation, while \(\Theta\)
characterizes the geometric reorientation of the optimal purification
frames. For the imperfect NOT channel, this distinction is physically
meaningful, i.e., the noise is not purely dissipative, but contains a coherent
rotation error associated with the imperfect pulse. The Fano--Procrustes
framework therefore gives a refined state-dependent characterization of
the gate, separating scalar distinguishability from geometric
misalignment.


\section{Analysis and discussion}
\label{sec:analysis_discussion}

\subsection{Fano--Procrustes approach}

The results presented in this work show that the Fano--Procrustes construction
provides a useful way of separating two different aspects of the
comparison between quantum states. The first aspect is scalar distinguishability,
which is quantified in this case by the metric \(D_N\). The second
is geometric alignment, quantified by the
optimal Procrustes rotation \(S_\star\) and by its associated angle
\(\Theta\). The central point is that these two quantities address different
questions. The metric \(D_N\) measures how close two states can
be made at the level of optimal purification overlap, whereas
\(\Theta\) describes how the corresponding purification frames must be
rotated in order to attain that optimal overlap.

This distinction is important because the scalar metric is determined
only by the value of the maximized overlap. Equivalently, it is directly
related to the fidelity between the two mixed states. Therefore, once
the maximal overlap \(g_\star\) is known, the value of \(D_N\) is fixed.
The Procrustes formulation reproduces this scalar information, but it
also gives access to the optimizer \(S_\star\). The angle \(\Theta\) is
extracted from this optimizer and therefore contains information about
the structure of the maximizing purification, rather than only about the
value of the maximum. In this sense, the Fano--Procrustes framework
upgrades the fidelity-based comparison from a scalar optimization
problem to a structured geometric problem consisting of both an optimal
value and an optimal transformation.

The numerical agreement between the metric \(D_N\) computed from the
Procrustes procedure and the same metric computed from the Uhlmann
fidelity served as a consistency check of the method.

The depolarizing channel provides the cleanest benchmark for this
interpretation. Since the channel action is purely radial, it changes
only the length of the Bloch vector, leaving its direction unchanged.
Consequently, the input and output Bloch vectors are always parallel,
the canonical Procrustes matrix is
diagonal in the adapted frame, and the optimal rotation is
\(S_\star=\I_3\). This explains why the misalignment angle vanishes for
all values of \(p\), while the metric $D_N$ remains nonzero whenever
the purity of the state is changed. The depolarizing channel thus
illustrates the separation between radial distinguishability and
geometric misalignment, i.e., \(D_N\) detects the contraction of the Bloch
vector, whereas \(\Theta\) correctly indicates the absence of
purification-frame twisting.

The bit-flip and phase-flip channels display a different behavior. Both
channels are unital, but they are anisotropic, i.e., each one preserves one
axis of the Bloch ball while contracting the transverse plane. For the
bit-flip channel, the \(x\)-component is invariant and the \(y\) and
\(z\) components are multiplied by \((1-2p)\). For the phase-flip channel,
the \(z\)-component is invariant and the transverse components are
multiplied by the same factor. This anisotropy implies that
axis-adapted states evolve collinearly and therefore have
\(\Theta=0\), while generic input states experience a change in Bloch
direction and consequently acquire a nontrivial Procrustes angle \(\Theta\).

The behavior of \(D_N\) for these channels reflects the total scalar
change in the state, including the loss or reversal of transverse
components. By contrast, the behavior of \(\Theta\) is controlled by
whether the channel changes the direction of the Bloch vector relative
to the chosen canonical frame. This explains the strong dependence of
\(\Theta\) on the input orientation. States close to the symmetry axis
of the channel show little or no misalignment, whereas states with
substantial transverse components generally require a nontrivial
optimal rotation. Thus, the bit-flip and phase-flip examples show that
\(\Theta\) is sensitive to anisotropic channel action, even when the
scalar metric alone does not explicitly identify the directional
origin of the change.

The amplitude damping channel adds another important feature, i.e., it is
nonunital. In Bloch form, it contracts the transverse components and
shifts the longitudinal component toward the ground state. Therefore,
the output Bloch vector is not obtained merely by an anisotropic
contraction about the origin, but also by an affine displacement. For
states diagonal in the energy basis, the input and output remain
collinear, so the Procrustes rotation is trivial and
\(\Theta=0\). For generic states, however, the affine drift changes the
direction of the Bloch vector. The resulting nonzero values of
\(\Theta\) capture the geometric contribution of this nonradial action.

This makes amplitude damping a particularly useful example for the
present framework. The metric \(D_N\) quantifies the total
effect of dissipative relaxation on the state. The angle \(\Theta\), on
the other hand, identifies whether this relaxation includes a geometric
reorientation of the purification frame. Hence, for amplitude damping,
the pair \((D_N,\Theta)\) separates the radial effect associated with
purity change from the directional effect associated with the affine
drift toward the fixed point of the channel.

The imperfect quantum NOT gate provides a complementary example because
its noise is not purely dissipative. Instead, it models a coherent
control error, i.e., with probability \((1-p)\) the ideal NOT operation is
applied, while with probability \(p\) the system undergoes a slightly
miscalibrated rotation about the same axis. The corresponding Bloch map
leaves the \(x\)-component invariant but mixes the \(y\) and \(z\)
components. This produces both scalar distinguishability and geometric
misalignment for generic input states.

The results for the imperfect NOT gate also illustrate the usefulness
of the geometric information contained in \(\Theta\). The metric $D_N$
measures how far the implemented output is from either the input state
or the target output. The angle
\(\Theta\) indicates whether the deviation involves a nontrivial
rotation of the optimal purification frame. Since the error in this
model is coherent and directional, \(\Theta\) provides information that
is naturally aligned with the physical origin of the imperfection. It
therefore complements the scalar distance by identifying the geometric
component of the gate error.

Across all examples, a common pattern emerges. When a channel acts
radially in the Bloch ball, changing only the norm of the Bloch vector,
the metric $D_N$ may be nonzero but the Procrustes angle \(\Theta\) vanishes.
When the channel changes the Bloch direction, either through anisotropic
contraction, affine drift, or coherent rotation error, the Procrustes
angle generally becomes nonzero. This supports the interpretation of
\(\Theta\) as a diagnostic of nonradial channel action in purification
space.

The role of degeneracies should also be kept in mind. In highly
symmetric situations, such as collinear Bloch-vector families, the
Procrustes problem can have degeneracies because the singular values of
the matrix \(K\) may coincide. These degeneracies do not affect the
 metric $D_N$, since all optimal rotations yield the same maximal
overlap. In many symmetry-adapted cases they are also benign for
\(\Theta\), because all relevant optimal choices lead to the same trace
of \(S_\star\), and hence to the same angle.

\subsection{Operational realization of the optimal purification}
\label{ssec:rodrigues_operational}

The geometric freedom underlying the Fano--Procrustes parametrization is
closely related to the operational purification scheme discussed by
Bassi and Ghirardi~\cite{Bassi2003}. In their approach, a purified
bipartite state is regarded as a physical resource from which different
ensemble decompositions of a given mixed state can be generated by
performing suitable measurements on the ancillary system. In the present
qubit setting, the ancilla-side freedom is represented
geometrically by rotations acting on the Fano variables of the purification.

The Procrustes optimization determines the optimal proper
rotation \(S_\star\in SO(3)\) from the singular value decomposition of
the Procrustes matrix \(K\). Its
axis--angle representation is

\begin{equation*}
S_\star
=
R(\hat{\mathbf u}_\star,\theta_\star),
\end{equation*}

\noindent where \(\hat{\mathbf u}_\star\in\mathbb R^3\) is a unit vector and
\(\theta_\star\in[0,\pi]\). In Rodrigues form~\cite{Rodrigues1840,MurrayLiSastry1994},

\begin{equation*}
R(\hat{\mathbf u}_\star,\theta_\star)
=
\I_3
+
\sin\theta_\star\,[\hat{\mathbf u}_\star]_\times
+
(1-\cos\theta_\star)\,[\hat{\mathbf u}_\star]_\times^2,
\end{equation*}

\noindent with

\begin{equation*}
[\hat{\mathbf u}_\star]_\times
=
\begin{pmatrix}
0 & -u_{\star,z} & u_{\star,y}\\
u_{\star,z} & 0 & -u_{\star,x}\\
-u_{\star,y} & u_{\star,x} & 0
\end{pmatrix}.
\end{equation*}

The rotation angle is obtained from

\begin{equation*}
\theta_\star
=
\arccos\!\left(
\frac{\Tr(S_\star)-1}{2}
\right).
\end{equation*}

For \(0<\theta_\star<\pi\), the corresponding unit axis is obtained from
the antisymmetric part of \(S_\star\):

\begin{equation}
\hat{\mathbf u}_\star
=
\frac{1}{2\sin\theta_\star}
\begin{pmatrix}
(S_\star)_{32}-(S_\star)_{23}\\
(S_\star)_{13}-(S_\star)_{31}\\
(S_\star)_{21}-(S_\star)_{12}
\end{pmatrix}.
\label{eq:axis_star_from_Sstar}
\end{equation}

The ancilla unitary implementing the optimal purification is the
\(SU(2)\) lift of this rotation. With the convention

\begin{equation*}
U_\star^\dagger
(\mathbf v\cdot\boldsymbol{\sigma})
U_\star
=
(S_\star\mathbf v)\cdot\boldsymbol{\sigma},
\qquad
\mathbf v\in\mathbb R^3,
\end{equation*}

\noindent one may choose

\begin{equation*}
U_\star
=
\exp\!\left[
-\frac{i}{2}\theta_\star\,
\hat{\mathbf u}_\star\cdot\boldsymbol{\sigma}
\right],
\end{equation*}

\noindent or equivalently

\begin{equation*}
U_\star
=
\cos\frac{\theta_\star}{2}\,\I
-
i\sin\frac{\theta_\star}{2}\,
\hat{\mathbf u}_\star\cdot\boldsymbol{\sigma}.
\end{equation*}

Applying this unitary to the ancillary subsystem of a reference
purification gives the Procrustes-optimal purification,

\begin{equation*}
|\Psi_\sigma^{(\star)}\rangle
=
(\I\otimes U_\star)
|\Psi_\sigma^{(e)}\rangle .
\end{equation*}

The limiting cases require separate but simple treatments. If
\(\theta_\star=0\), then \(S_\star=\I_3\), and one may choose

\begin{equation*}
U_\star=\I.
\end{equation*}

If \(\theta_\star=\pi\), Eq.~\eqref{eq:axis_star_from_Sstar} is singular.
In that case, the axis \(\hat{\mathbf u}_\star\) is obtained as the
normalized eigenvector of \(S_\star\) with eigenvalue \(+1\). The
corresponding unitary may be chosen as

\begin{equation*}
U_\star
=
-i\,\hat{\mathbf u}_\star\cdot\boldsymbol{\sigma},
\end{equation*}

\noindent up to an irrelevant overall sign.

Finally, the lift from \(SO(3)\) to \(SU(2)\) is two-to-one, i.e., both
\(U_\star\) and \(-U_\star\) generate the same rotation \(S_\star\) and
therefore the same transformed purification projector. Thus the
Rodrigues representation provides the practical link between the
geometric output of the Procrustes algorithm and the ancilla-side
unitary operation needed to realize the optimal purification.
This establishes a direct connection between the Fano--Procrustes
optimization and the operational purification viewpoint. Indeed, in the qubit
setting, the freedom to act on the ancilla generates the family of
purifications, while the Procrustes solution selects the ancilla
operation that produces the overlap-optimal purification within the
canonical gauge.

\section{Concluding remarks}
\label{concluding_remarks}

In this work we reformulated the Uhlmann
purification-overlap optimization and developed a purification-based
geometric framework for the analysis of mixed qubit states and qubit
channels, combining the entropic metric \(D_N\) with an explicit
optimization over purifications written in Fano form. The central idea
of the paper is that, for one-qubit mixed states, the freedom in the
choice of purification can be represented in a finite-dimensional and
geometrically transparent way, so that the optimization problem
underlying the metric \(D_N\) can be reduced to an orthogonal
Procrustes problem on \(SO(3)\).

The \(D_N\) metric adopted in this work was originally defined for pure
states and later extended to mixed states through an optimization over
all purifications. By Uhlmann's theorem, this
optimization is equivalent to maximizing the overlap between purifications,
which in turn is equivalent to the square root of the Uhlmann fidelity.
Although this already provides a closed formula for the scalar value of the
metric, the main purpose of the present work has not been limited to the
evaluation of \(D_N\) itself. Rather, our goal has been to analyze the
structure of the optimizing purifications and to extract from that
optimization additional geometric information.

To that end, we adopted the Fano representation of two-qubit pure states,
expressing a purification in terms of the Bloch vector \(\mathbf r\) of the
reduced state, the ancillary Bloch vector \(\boldsymbol{\gamma}\), and a
real \(3\times 3\) matrix \(A\) collecting the correlations between Pauli
observables on the system and the ancilla. The purity of the bipartite
state imposes strong algebraic restrictions on these quantities, in
particular on the matrix \(A\). These restrictions show that, for a fixed
mixed qubit state, the allowed purifications form a structured family that
can be generated by right multiplication of a reference matrix by elements
of \(SO(3)\), together with the corresponding induced transformation of the
ancilla Bloch vector.

This observation leads directly to the key reduction of the paper. The
overlap between two purifications can be written explicitly in terms of
their Fano parameters, and the optimization over purifications can be recast as
the maximization of \(\Tr(KS)\) over \(S\in SO(3)\), where \(K\) is a
matrix constructed from the correlation matrices and ancillary Bloch
vectors of the two states. In other words, the optimization over
purifications becomes an orthogonal Procrustes problem. This is
a central result of the work, i.e., a formally infinite-dimensional
optimization over purifications is replaced by an explicit and tractable
optimization over proper rotations.

A further contribution of the paper is the construction of a canonical
gauge adapted to the qubit case. By introducing a convenient reference
frame tied to the Bloch vector of the state, we obtained explicit forms for
the purification data and for the Procrustes matrix \(K\). This allowed us
to derive compact expressions for the optimization problem, and to make precise the geometric meaning of the
ancilla-side freedom. In this canonical gauge, the orthogonal matrix
entering the Procrustes optimization plays the role of the Bloch-sphere
representation of a unitary acting on the ancilla qubit.

This geometric construction naturally led to the definition of the
purification misalignment angle \(\Theta\), extracted from the optimal
Procrustes rotation \(S_\star\). While \(D_N\) depends only
on the optimized overlap and therefore captures the scalar
distinguishability between the states, \(\Theta\) retains directional
information about how the corresponding optimal purification frames are
related. In this sense, \(\Theta\) provides a geometric indicator of the
ancilla-side realignment required to achieve the maximal overlap between
purifications.

We also investigated several analytically tractable classes of quantum
channels. For depolarizing, bit-flip, phase-flip, and amplitude-damping
channels restricted to symmetry-adapted families of input states, the input
and output Bloch vectors remain collinear. In these cases, the Procrustes
optimization becomes trivial, the optimal rotation is the identity, and the
misalignment angle \(\Theta\) vanishes. The overlap and the corresponding 
metric $D_N$ can then be computed in closed form. These families provide useful
benchmark cases, showing that the framework reproduces the expected
behavior in situations where the dynamics is purely radial or preserves the
relevant Bloch direction.

The numerical analysis for more general channels complements the analytical
results by showing how nontrivial values of \(\Theta\) arise when the
channel action changes the direction of the Bloch vector. In particular,
the treatment of the imperfect quantum NOT gate illustrates how the
present framework can be used to study coherent control errors. In that
case, the Procrustes angle captures the geometric mismatch between the
optimal purification frames, while the metric $D_N$ quantifies the
corresponding loss of overlap. This makes the pair \((D_N,\Theta)\)
especially suggestive as a state-dependent diagnostic of channel action, i.e., 
the first quantity measures the magnitude of the change, while the second
reveals whether that change includes a nontrivial geometric reorientation
in purification space.

Another result of conceptual relevance is the operational interpretation of
the Procrustes optimization. Because every proper rotation in \(SO(3)\)
arises as the Bloch-sphere representation of a unitary in \(SU(2)\), the
optimal Procrustes rotation \(S_\star\) can be lifted to a unitary
\(U_\star\) acting on the ancilla qubit. Consequently, once a reference
purification is prepared, the optimal purification singled out by the
Fano--Procrustes construction can be physically generated by a local
ancilla-side unitary transformation. In this way, the optimization over
purifications is not merely an abstract mathematical procedure, i.e., it can be
understood as selecting the ancilla operation that produces the
overlap-optimal purification. This provides a direct bridge between the
present geometric framework and operational purification schemes based on
ancilla interactions and measurements \cite{Bassi2003}.

Taken together, these results show that the Fano--Procrustes approach offers
a unified perspective on purification geometry, metric distinguishability,
and channel analysis in the qubit setting. The framework makes explicit the
ancilla-side structure hidden behind the overlap optimization, provides a
tractable method to construct the optimal purification, and reveals
geometric information beyond that contained in a scalar metric alone. It
also suggests a natural separation between radial and nonradial aspects of
quantum channel action, with \(D_N\) quantifying the former and
\(\Theta\) offering information about the latter.

There are several natural directions for further work. One is the extension
of the present construction beyond qubits, where the purification freedom
is no longer described by simple rotations and the correlation structure
becomes significantly richer. Another is the systematic study of the
relation between the present geometric quantities and other state-space
metrics or channel diagnostics, particularly those arising in quantum
information geometry and quantum control. Finally, the operational
interpretation of the Procrustes optimization suggests possible
applications in ancilla-assisted state engineering, remote ensemble
preparation, and noise characterization, where the optimal purification may
have direct physical relevance beyond its role in the definition of the
 metric $D_N$.
 
It is important to emphasize that the present construction is not tied exclusively to the metric
\(D_N\). Any distance or distinguishability measure that is a function
of the Uhlmann fidelity, or equivalently of the maximal purification
overlap, can be evaluated using the same Procrustes optimization. In
such cases, the scalar distance is obtained by applying the appropriate
function to \(g_\star\). The additional advantage of the
Fano--Procrustes formulation is that it also identifies the optimizing
rotation \(S_\star\). Therefore, while different fidelity-based
distances provide different scalar measures of distinguishability, the
quantities \(S_\star\) and \(\Theta\) provide a complementary geometric
description of how the optimal purifications are aligned. This geometric feature is invisible
to purely fidelity-based quantities.

\begin{acknowledgments}
The author is grateful to Dr.~Pedro W.~Lamberti for numerous
enlightening discussions related to this work. The author also
acknowledges financial support from CONICET (PIP No.~135/23) and from
the Secretaría de Ciencia y Técnica de la Universidad Nacional de
Córdoba (SeCyT--UNC, Argentina, Proyecto Consolidar 2023--2025,
Director: Pedro W.~Lamberti).
\end{acknowledgments}

\bigskip

\appendix

\section{Uniqueness of $g_\star$}

\begin{lemma}[Uniqueness of the optimal overlap]
Let
\[
g^2(S)=\frac14\bigl(1+\mathbf r\cdot \mathbf s+\Tr(KS)\bigr),
\qquad S\in SO(3),
\]
and define
\[
g_\star=\max_{S\in SO(3)} g(S).
\]
Then the value \(g_\star\) is uniquely determined.
\end{lemma}

\begin{proof}
Let \(K=U\Sigma V^T\) be a singular value decomposition, with
\(\Sigma=\diag(\sigma_1,\sigma_2,\sigma_3)\), \(\sigma_1\ge\sigma_2\ge\sigma_3\ge 0\).
For any \(S\in SO(3)\), set \(Q=V^T S U\), so that \(Q\in SO(3)\) and
\[
\Tr(KS)=\Tr(\Sigma Q).
\]
Hence
\[
\max_{S\in SO(3)}\Tr(KS)=\max_{Q\in SO(3)}\Tr(\Sigma Q),
\]
whose value is
\[
\sigma_1+\sigma_2+\operatorname{sgn}(\det(VU^T))\,\sigma_3.
\]
This number depends only on the singular values of \(K\), and is therefore unique.
It follows that \(g_\star^2\), and hence \(g_\star\), is uniquely determined.
\end{proof}

\section{optimal overlap for commuting qubit states}
\label{app:collinear}

Let
\[
\rho_a=\frac12\left(\I+a\,\hat{\mathbf n}\cdot\boldsymbol{\sigma}\right),
\qquad
\sigma_b=\frac12\left(\I+b\,\hat{\mathbf n}\cdot\boldsymbol{\sigma}\right),
\]
with \(a,b\in[-1,1]\). Then the optimal purification overlap is
\[
g_\star(\rho_a,\sigma_b)
=
\frac12
\left[
\sqrt{(1+a)(1+b)}
+
\sqrt{(1-a)(1-b)}
\right].
\]

Indeed, since \(\rho_a\) and \(\sigma_b\) are diagonal in the common eigenbasis

\(\{\ket{+_{\hat n}},\ket{-_{\hat n}}\}\), we write
\[
\rho_a
=
\lambda_+\ket{+_{\hat n}}\bra{+_{\hat n}}
+
\lambda_-\ket{-_{\hat n}}\bra{-_{\hat n}},
\]

\noindent and

\[
\sigma_b
=
\mu_+\ket{+_{\hat n}}\bra{+_{\hat n}}
+
\mu_-\ket{-_{\hat n}}\bra{-_{\hat n}},
\]

\noindent where

\[
\lambda_\pm=\frac{1\pm a}{2},
\qquad
\mu_\pm=\frac{1\pm b}{2}.
\]

Choose the purifications

\[
\ket{\Psi_\rho}
=
\sqrt{\lambda_+}\ket{+_{\hat n},+_{\hat n}}
+
\sqrt{\lambda_-}\ket{-_{\hat n},-_{\hat n}},
\]

\noindent and

\[
\ket{\Psi_\sigma}
=
\sqrt{\mu_+}\ket{+_{\hat n},+_{\hat n}}
+
\sqrt{\mu_-}\ket{-_{\hat n},-_{\hat n}}.
\]

Their overlap is

\[
\braket{\Psi_\rho}{\Psi_\sigma}
=
\sqrt{\lambda_+\mu_+}
+
\sqrt{\lambda_-\mu_-}.
\]

It remains to show that $|\braket{\Psi_\rho}{\Psi_\sigma}|$ is maximal. Any other purification of
\(\sigma_b\) can be written as~\cite{Jozsa1994,Nielsen2000,Bengtsson2017}

\[
(\I\otimes U)\ket{\Psi_\sigma},
\]

\noindent where \(U\) is unitary on the ancillary Hilbert space\cite{Jozsa1994,Nielsen2000,Bengtsson2017}. Therefore

\[
\braket{\Psi_\rho}{(\I\otimes U)\Psi_\sigma}
=
\sqrt{\lambda_+\mu_+}\,\bra{+_{\hat n}}U\ket{+_{\hat n}}
+
\sqrt{\lambda_-\mu_-}\,\bra{-_{\hat n}}U\ket{-_{\hat n}}.
\]

Taking the absolute value and using

\[
\left|\bra{\pm_{\hat n}}U\ket{\pm_{\hat n}}\right|\leq 1,
\]
\noindent we obtain

\[
\left|
\braket{\Psi_\rho}{(\I\otimes U)\Psi_\sigma}
\right|
\le
\sqrt{\lambda_+\mu_+}
+
\sqrt{\lambda_-\mu_-}.
\]

The upper bound is attained by choosing \(U=\I\). Hence the matched
Schmidt purifications already maximize the overlap.

Substituting the eigenvalues gives

\[
g_\star(\rho_a,\sigma_b)
=
\sqrt{\frac{1+a}{2}\frac{1+b}{2}}
+
\sqrt{\frac{1-a}{2}\frac{1-b}{2}},
\]

\noindent or equivalently

\[
g_\star(\rho_a,\sigma_b)
=
\frac12
\left[
\sqrt{(1+a)(1+b)}
+
\sqrt{(1-a)(1-b)}
\right].
\]

\section{Properties of the purification-based distance \(D_N\)}

Let \(\mathcal{D}(\mathcal{H})\) be the set of density operators acting on a
finite-dimensional Hilbert space \(\mathcal{H}\). For all \(\rho,\sigma,\tau\in\mathcal{D}(\mathcal{H})\), the distance \(D_N\)
satisfies the following properties~\cite{Lamberti2009,Osan2013}.

\paragraph{Normalization.}

\[
0\leq D_N(\rho,\sigma)\leq 1.
\]

\paragraph{Identity of indiscernibles.}

\[
D_N(\rho,\sigma)=0
\quad\Longleftrightarrow\quad
\rho=\sigma.
\]

\paragraph{Maximal distance for orthogonal supports.}

\[
D_N(\rho,\sigma)=1
\quad\Longleftrightarrow\quad
F(\rho,\sigma)=0.
\]

Equivalently, \(D_N(\rho,\sigma)=1\) if and only if the supports of
\(\rho\) and \(\sigma\) are orthogonal.

\paragraph{Symmetry.}

\[
D_N(\rho,\sigma)=D_N(\sigma,\rho).
\]

\paragraph{Triangle inequality.}

\[
D_N(\rho,\sigma)
\leq
D_N(\rho,\tau)+D_N(\tau,\sigma).
\]

Consequently, \(D_N\) is a metric on \(\mathcal{D}(\mathcal{H})\).

\paragraph{Unitary invariance.}

For every unitary operator \(U\),
\[
D_N(U\rho U^\dagger,U\sigma U^\dagger)
=
D_N(\rho,\sigma).
\]

\paragraph{Invariance under appending an identical ancillary state.}
For every ancillary density operator \(\omega\),

\[
D_N(\rho\otimes\omega,\sigma\otimes\omega)
=
D_N(\rho,\sigma).
\]

\paragraph{Contractivity under CPTP maps.}
For every completely positive trace-preserving map \(\Phi\),

\[
D_N(\Phi(\rho),\Phi(\sigma))
\leq
D_N(\rho,\sigma).
\]

In particular, for partial trace,

\[
D_N(\operatorname{Tr}_B\rho_{AB},\operatorname{Tr}_B\sigma_{AB})
\leq
D_N(\rho_{AB},\sigma_{AB}),
\]

and for a nonselective projective measurement \(\{P_i\}_i\),

\[
D_N\!\left(
\sum_i P_i\rho P_i,
\sum_i P_i\sigma P_i
\right)
\leq
D_N(\rho,\sigma).
\]

\paragraph{Pure-state reduction.}
If

\[
\rho=|\psi\rangle\langle\psi|,
\qquad
\sigma=|\phi\rangle\langle\phi|,
\]

\noindent then

\[
D_N^2(\rho,\sigma)
=
S\!\left(
\frac{
|\psi\rangle\langle\psi|
+
|\phi\rangle\langle\phi|
}{2}
\right),
\]

\noindent where \(S(\omega)=-\operatorname{Tr}(\omega\log_2\omega)\). Equivalently,

\[
D_N^2(\rho,\sigma)
=
\Phi_N(|\langle\psi|\phi\rangle|).
\]

For pure states, \(D_N^2\) coincides with the quantum Jensen--Shannon
divergence~\cite{Majtey2005}:

\[
D_N^2
\bigl(
|\psi\rangle\langle\psi|,
|\phi\rangle\langle\phi|
\bigr)
=
D_{\mathrm{JS}}
\bigl(
|\psi\rangle\langle\psi|,
|\phi\rangle\langle\phi|
\bigr).
\]


\end{document}